\documentclass[preprint,3p]{elsarticle}
\makeatletter
\def\ps@pprintTitle{%
 \let\@oddhead\@empty
 \let\@evenhead\@empty
 \let\@evenfoot\@oddfoot}
\makeatother

\usepackage[colorlinks,
    linkcolor={red!50!black},
    citecolor={red!50!black},
    urlcolor={blue!80!black},
    pdfauthor={Author Names}]{hyperref}
   
\usepackage[colorinlistoftodos]{todonotes}

\newcommand{\x}{\textnormal{\textbf{x}}}

\usepackage{amssymb}
\usepackage{amsmath}
\usepackage{cleveref}
\usepackage{amsfonts}
\usepackage{amsthm}
\usepackage{mathtools}
\usepackage{amssymb}
\usepackage{comment}
\usepackage{graphicx}
\usepackage[utf8]{inputenc}
\usepackage{appendix}
\usepackage{booktabs}
\DeclareGraphicsExtensions{.pdf,.png,.jpg,.eps}
\usepackage{algorithm}
\usepackage{algpseudocode}
\newtheorem{remark}{Remark}
\newtheorem{theorem}{Theorem}

\newtheorem{corollary}{Corollary}

\newtheorem{definition}{Definition}
\newtheorem{prop}[theorem]{Proposition}

\Crefname{equation}{Eq.}{Eqs.}
\Crefname{prop}{Proposition}
{Propositions}
\Crefname{figure}{Figure}{Figures}
\begin{document}

\begin{frontmatter}



\title{A self-supervised learning approach for denoising autoregressive models with additive noise: finite and infinite variance cases
}


\author[inst1]{Sayantan Banerjee\corref{cor1}}
\ead{sayantanbanerjee74@gmail.com}
\author[inst2]{Agnieszka Wy{\l}oma{\'n}ska}
\ead{agnieszka.wylomanska@pwr.edu.pl}
\author[inst1]{S.Sundar}
\ead{slnt@iitm.ac.in}
\cortext[cor1]{Corresponding author}
\affiliation[inst1]{organization={Department of Mathematics},%Department and Organization
            addressline={Indian Institute of Technology Madras}, 
            city={Chennai},
            postcode={600036}, 
            state={Tamil Nadu},
            country={India}}

%%\author[inst2]{Agnieszka Wy{\l}oma{\'n}ska}

\affiliation[inst2]{organization={Faculty of Pure and Applied Mathematics, Hugo Steinhaus Center},%Department and Organization
            addressline={Wroclaw University of Science and Technology, Wyspianskiego 27}, 
            %%city={},
            postcode={50-370}, 
            state={Wroclaw},
            country={Poland}}

\begin{abstract}
 The autoregressive time series model is a popular second-order stationary process, modeling a wide range of real phenomena.  However, in applications, autoregressive signals are often corrupted by additive noise. Further, the autoregressive process and the corruptive noise may be highly impulsive, stemming from an infinite-variance distribution. The model estimation techniques that account for additional noise tend to show reduced efficacy when there is very strong noise present in the data, especially when the noise is heavy-tailed. In this paper, we propose a novel self-supervised learning method to denoise the additive noise-corrupted autoregressive model. Our approach is motivated by recent work in computer vision and does not require full knowledge of the noise distribution. We use the proposed method to recover exemplary finite- and infinite-variance autoregressive signals, namely, Gaussian and alpha-stable distributed signals, respectively, from their noise-corrupted versions. The simulation study conducted on both synthetic and semi-synthetic data demonstrates strong denoising performance of our method compared to several baseline methods, particularly when the corruption is significant and impulsive in nature. Finally,  we apply the presented methodology to forecast the pure autoregressive signal from the noise-corrupted data.
 
\end{abstract}


\end{frontmatter}

%\newpage
\section{Introduction}\label{sec:Introduction}
The autoregressive (AR) time series model is widely used to model data arising from different areas such as signal processing \cite {kay1988modern, Legrand2018, Ganapathy2017}, hydrology \cite{Lohani2012}, river flow forecasting \cite{Terzi2013}, and finance \cite{Marcellino2006, Weron2008}. The classical notions dictate the model to be noise-free and governed by probability distributions with finite variance. However, these assumptions may result in the model being inadequate to capture data in the presence of large fluctuations. The $\alpha$-stable time series model is a meaningful choice when the data are impulsive and non-Gaussian in nature (see \cite{Kabasinskas2009641, Nolan2003}). The usage of $\alpha$-stable distribution in such cases stems from the generalized Central Limit Theorem (GCLT), which states that $\alpha$-stable distribution is the only possible limit distribution for normalized sums of independent and identically distributed (i.i.d.) random variables with infinite variance. This characterization enables $\alpha$-stable distribution to adequately capture data with extreme values and high variability. Thus, $\alpha$-stable distribution has found application in various fields, more significantly in signal processing \cite{Nikias1994}, finance \cite{McCulloch1996} and condition monitoring \cite{zak2017}. A recent publication \cite{Yuan2025} has used the impulsive nature of $\alpha$-stable distribution when training a neural network to enhance its robustness.\\
In this paper, we consider the AR time series model with additive noise, where both the AR model and the noise may have infinite variance. These considerations are important because real-world signals are often impulsive and corrupted by additional noise. When additive noise is present in the model, classical approaches in model estimation and identification are not sufficient. Novel methods have been introduced to these ends, accommodating additional noise \cite{Diversi2005, Diversi2007, Esfandiari2020, Zulawinski2023, Zulawinski2023a}. When additive noise is considered to have infinite variance, due to GCLT, the $\alpha$-stable distribution is taken as the exemplary distribution in \cite{Zulawinski2023b, Zulawinski2024} while estimating the noise-corrupted model. However, robust estimation techniques, while taking additive noise into account, may yield large errors for very strong noise. Furthermore, the problem of forecasting pure data from noise-corrupted samples may require the development of additional techniques. Therefore, recovering the pure autoregressive model by denoising the noise-corrupted signal is crucial for estimating and forecasting the model, particularly in the presence of highly impulsive noise.\\
The problem of denoising a real-world signal has been addressed through several approaches so far, including empirical mode decomposition and dictionary learning \cite{Li2017, Wu2018}. Several filtering-based methods have also been proposed \cite{Sameni2007, Gao2010, Chen2008}. More recently, robust methods have been developed using Kalman filters for heavy-tailed distributions \cite{Hao2023}. Wavelet-based methods \cite{Bayer2019} have also found broader applications in the context of time series denoising. These include approaches such as \cite{Gharesi2020}, with applications to condition monitoring and finance \cite{Alrumaih2002}. Furthermore, a combination of wavelet denoising and supervised learning through a convolutional autoencoder has been proposed \cite{Frusque2024} to achieve better generalization across different noise levels.
\\
Denoising methods, through both supervised and self-supervised learning, have been frequently addressed in works related to computer vision. The removal of noise in an image \cite{Fan2019, Vasilyeva2025} is an important issue in this context. Denoising networks have been successfully applied to various image denoising problems. Supervised networks train by mapping a noisy image to a clean image. In general, architectures based on convolutional neural network (CNN) \cite{Zhang2017} or multilayer perceptron (MLP) \cite{Mansour2022, Tu_2022_CVPR} have achieved state-of-the-art performance through supervised learning. The Noise2Noise method, introduced by \cite{lehtinen2018}, enables learning to denoise without requiring ground-truth images. Under the assumption that the noise has zero mean, Noise2Noise works by mapping one noisy version of the underlying pure image to another noisy version of the same image, bypassing the need for pure images at the training stage. Noise2Noise networks have shown strong denoising performance. Although it is hard to procure two noisy copies of the same image in practical scenarios, Noise2Noise has inspired a plethora of self-supervised methods, which require only a single version of the noisy image to train. However, self-supervised methods such as Noisy-As-Clean (NAC) \cite {Xu2020} or Noisier2Noise (NR2N) \cite{Moran2020} require full knowledge of the noise distribution. While under the assumption of zero-mean noise, \cite{Diversi2007} provides complete knowledge of the noise distribution for a noise-corrupted AR model, where both the AR model and the noise are Gaussian-distributed, assuming prior knowledge about the additive noise is not realistic when either of them may be impulsive.\\
In this paper, motivated by the Noise2Noise method, we propose a novel self-supervised learning method, Stable-Noise2Noise (Stable-N2N) to recover the pure symmetric $\alpha$-stable AR ($S\alpha S$-AR) time series data, where $\alpha>1$, from its additive noise-corrupted version. Our method only takes into account the model assumptions made in \cite{Zulawinski2024}, which are intricate to the noise-corrupted model estimation. We do not require explicit knowledge of the noise model.\\
We propose to judge the denoising quality of our method in the following way. Firstly, we note that the classical low-order Yule-Walker (YW) method for the estimation of an AR model yields a biased estimation in the presence of additive noise \cite{Kay1980, Esfandiari2020}. Moreover, when the AR model is governed by $\alpha$-stable distribution, the modified low-order YW method, based on the notion of fractional lower-order covariance (FLOC) \cite{Zulawinski2022}, also leads to a biased solution when additional noise is present \cite{Zulawinski2024}. We apply the low-order YW method on the denoised data and measure the denoising quality through the bias associated with the model estimation. If the noise in the model is successfully removed, the low-order YW method should yield an unbiased estimation of the model parameters, when applied on denoised data. The mean absolute error (MAE) is computed between the true values of the parameters and the estimated parameters as a measure of bias in the solution of the low-order YW method. Finally, as an application to the proposed methodology, we construct a 5-steps ahead forecast of the pure $S\alpha S$-AR time series from its noise-corrupted version combining denoised data points and estimation of a noise-corrupted model with errors-in-variables (EIV) method \cite{Diversi2007, Zulawinski2024}. We show that denoising is crucial in the context of forecasting the pure AR time series model from the noise-corrupted data. Through extensive Monte Carlo simulations, we demonstrate the efficacy and robustness of the proposed approach compared to several baselines. As a particular case of the considered noise-corrupted  $S\alpha S$-AR time series model, we apply the proposed approach on a Gaussian AR model with Gaussian additive noise to accommodate the finite variance case. Although, as mentioned before, the Gaussian-distributed model provides complete knowledge of additive noise, fulfilling the requirements of the considered baseline methods \cite{Xu2020, Moran2020}, the proposed approach still outperforms them. Therefore, in the context of denoising noise-corrupted AR time series models, we believe that the proposed methodology will serve as an important technique. As our denoising network, we choose the feedforward neural network (FNN) consisting of fully connected layers because they have shown great performance in extracting patterns out of AR time series models \cite{Zhang2003, Khashei2011}, or more specifically $\alpha$-stable time series models consisting of strongly impulsive signals \cite{Sathe2023}.\\
The main contributions of this paper are as follows.
\begin{itemize}
\item Introduction of Stable-N2N : A novel self-supervised technique for noise-corrupted $S\alpha S$-AR time series denoising, Stable-N2N relies solely on the given noise-corrupted data while training, to remove noise. The proposed method does not require explicit knowledge of the corruption in the model, leading to robust denoising performance across various noise levels for both finite- and infinite-variance models.
\item Achieving superior denoising performance over considered baselines : Comparative analysis against the existing self-supervised denoising methods over synthetic and semi-synthetic datasets demonstrates the superiority of the proposed methodology.
\item Prediction of pure $S\alpha S$-AR time series data : The proposed method is applied in combination with the estimation of the noise-corrupted model to forecast the data points of the pure time series from the noise-corrupted trajectory.
\end{itemize}

The remainder of the paper is organized as follows. In \Cref{sec:AR with additive noise}, we introduce the noise-corrupted autoregressive model and the notions of alternative dependency measures for random variables with infinite variance. In \Cref{sec:method}, we describe the proposed methodology and the motivation behind it. \Cref{sec:simulation} contains a demonstration of the effectiveness of our method through Monte Carlo simulations. We remove Gaussian and $\alpha$-stable noise across a range of noise levels by applying the proposed method to synthetic and semi-synthetic trajectories of the noise-corrupted models. \Cref{sec:conclusion} summarizes the paper. In \ref{app:finite variance estimation} and \ref{app:estimation infinite variance}, we recall estimation methods of finite- and infinite-variance noise-corrupted models, respectively. \ref{app: sensitivity analysis} contains a sensitivity analysis on the performance of the proposed method with respect to a crucial hyperparameter. \ref{app:universality} showcases the universality of the method, and \ref{app:snr} contains analysis based on an alternative performance metric.

The following nomenclature is used throughout the article.
\begin{itemize}
  \renewcommand\labelitemi{} 
\item AR - autoregressive
\item i.i.d. - independent and identically distributed
\item FNN - feedforward neural network
\item NAC - Noisy-As-Clean
\item NR2N - Noisier2Noise
\item Stable-N2N - Stable-Noise2Noise
\item $S\alpha S$-AR - symmetric $\alpha$-stable autoregressive
\item YW - Yule-Walker
\item FLOC - fractional lower-order covariance
\item MAE - mean absolute error
\item EIV - errors-in-variables
\end{itemize}

\section{AR model with additive noise}
\label{sec:AR with additive noise}

\subsection{The \texorpdfstring{$\alpha$}{alpha}-stable distribution}
Univariate $\alpha$-stable distributions (also called L\'evy stable)  are a class of heavy-tailed distributions first introduced by Paul L\'evy in \cite{Levy1925}. The $\alpha$-stable distribution, denoted by $S(\alpha,\beta,\sigma,\mu)$ is characterized by four parameters, namely, stability index  $0 < \alpha \leq 2$, scale parameter  $\sigma >0$,  skewness parameter  $-1 \leq \beta \leq 1$ and shift parameter $\mu \in \mathbb{R} $. Out of these parameters, $\alpha$ determines the tail behavior of the distribution. As $\alpha$ decreases, the tail becomes heavier, leading to a more impulsive nature in the models governed by these distributions.   For $\alpha >1$, mean of the distribution exists and is equal to $\mu$. For more information on the $\alpha$-stable distribution, we refer the readers to the following bibliographic positions \cite{samorodnitsky1994stable,Nolan2020,borak2005stable}.\\
In general, probability densities are not in closed form for this family of distribution except for some particular cases such as Gaussian ($\alpha =2$, $\beta = 0$), L\'evy ($\alpha =1/2$, $\beta = 1$) or Cauchy ($\alpha = 1$, $\beta =0$) distribution. However, we can define $\alpha$-stable distribution in terms of characteristic function \cite{samorodnitsky1994stable}. We note that there are other equivalent ways available in the literature to define the $\alpha$-stable distribution.
\begin{definition}
A random variable $X$ is said to have the $\alpha$-stable distribution with parameters $\alpha \in (0,2]$, $\sigma > 0$, $\beta \in [-1,1]$, $\mu \in \mathbb{R}$ if its characteristic function has the following form
\begin{equation}
 E[\exp(itX)] = 
  \begin{cases}
\exp \{-\sigma^{\alpha}|t|^{\alpha} [1-i \beta \operatorname{sign}(t) \tan (\frac{\pi \alpha}{2})] +i \mu t\},&\hspace{1 mm} \alpha \neq 1 \\
\exp \{-\sigma|t| [1+i \beta \frac{2}{\pi} \operatorname{sign}(t) \ln|t|]  +i \mu t\}, &\hspace{1 mm} \alpha = 1.
\end{cases}
\end{equation}
\end{definition}
In particular, when $\beta =0$ and $\mu =0$, we call the distribution as symmetric $\alpha$-stable distribution ($S\alpha S$ distribution) and denote it by $S(\alpha,\sigma)$.
\begin{definition}
The random variable $X$ is called $S\alpha S$ with parameters $ \alpha \in (0,2] $, $\sigma > 0$ if the characteristic function takes the form 
\begin{equation}
E[exp(it X)] = exp\{-\sigma^\alpha|t|^\alpha\}.
\end{equation}
\end{definition}
For $\alpha = 2$, the $S\alpha S$ random variable $X$ is Gaussian distributed with mean zero. We denote it by $N(0,\nu)$ with variance $\nu$.
\subsection{\texorpdfstring{$S\alpha S$-AR time series model}{SαS-AR time series model}}
The classical second-order AR model is well known in the literature for modeling stationary time series with zero mean and finite second moment (variance) \cite{brockwell2002introduction}. However, this notion is extended in \cite{samorodnitsky1994stable, Gallagher2001, Kruczek2017} by considering a non-Gaussian model with infinite variance where the time series appears to be heavy-tailed in nature and has large values. To this end, $S\alpha S$-AR time series model is defined.
\begin{definition}
A time series $\{X_t\}_{t \in \mathbb{Z}}$ is called the one-dimensional $S\alpha S$-AR model of order $p$ if for each $t \in \mathbb{Z} $ it satisfies the following equation
\begin{equation}\label{eq:stable AR}
 X_t - \theta_1 X_{t-1} -\theta_2 X_{t-2} -\ldots-\theta_p X_{t-p}  = \xi_t, 
\end{equation}
where $\{\xi_t\}_ {t \in \mathbb{Z}}$ is a sequence of i.i.d. S$\alpha S$ random variables with $\alpha \in (1,2]$.

\end{definition}
We denote the one-dimensional $S\alpha S$-AR time series model of order $p$ as $S\alpha S$-AR(p) model. $\theta_1,\ldots,\theta_p$ are called the parameters of the model.\\
The bibliographic position \cite{samorodnitsky1994stable} provides the condition for the case when \Cref{eq:stable AR} admits a stationary solution.
\begin{theorem}
\Cref{eq:stable AR} has a unique solution of the form 
\begin{equation}\label{eq:causal}
 X_t = \sum_{j=0}^{\infty} c_j\xi_{t-j}, \qquad \text{t $\in \mathbb{Z}$} 
\end{equation}
 almost surely, \\
with real $c_j$'s satisfying $|c_j| < Q^{-j}$ eventually, $Q >1$, if and only if the polynomial $\theta (b) = 1- \theta_1 b-\ldots-\theta_{p} b^{p}$ has no root in the closed unit disk $\{b : |b| \leq 1\}$.
The sequence $\{X_t\}_{ t \in \mathbb{Z}}$ is then stationary and follows the $S\alpha S$ distribution.
\end{theorem}
 $\{X_t\}_{ t \in \mathbb{Z}}$ is said to be causal if it has a representation in the form of \Cref{eq:causal}.
\subsection{\texorpdfstring{$S\alpha S$-AR(p) time series model with $S\alpha S$ additive noise}{SαS-AR(p) time series model with SαS additive noise}}\label{subsec:noisy ar model}
We extend the $S\alpha S$-AR(p) time series by considering additional noise in the model. $\{X_t\}_{ t \in \mathbb{Z}}$ is assumed to be causal with the i.i.d. innovation sequence $\{\xi_t\}_ {t \in \mathbb{Z}}$ taken from the distribution $S(\alpha_\xi, \sigma_\xi)$, $\alpha_\xi > 1$. In addition, i.i.d. additive noise $\{Z_t\}_{t \in \mathbb{Z}}$ is taken from the distribution $S(\alpha_z, \sigma_z)$, $\alpha_z > 1$, thus generalizing the model considered in \cite{Diversi2005, Diversi2007, Esfandiari2020} by incorporating the presence of heavy-tailed behavior in autoregression and additive noise. We note that the model considered in this paper is a special case of the infinite-variance periodic AR (PAR) model with additive noise, considered in \cite{Zulawinski2024}.\\
The $S\alpha S$-AR(p) time series model with $S\alpha S$ additive noise, $\{Y_t\}_{t \in \mathbb{Z}}$ is defined as
\begin{equation}\label{eq:noise-corrupted AR model}
     Y_t = X_t + Z_t.
\end{equation}
We assume that $\{X_t\}_{t \in \mathbb{Z}}$ and the i.i.d. sequence $\{Z_t\}_{t \in \mathbb{Z}}$ are independent.\\
One can show that \Cref{eq:noise-corrupted AR model} can alternatively be written as 
\begin{equation}
    Y_t - \theta_1 Y_{t-1} -\ldots- \theta_p Y_{t-p} = \xi_t + Z_t -\theta_1 Z_{t-1} -\ldots- \theta_p Z_{t-p}.
\end{equation}
Let us note that the AR time series and additive noise may have different stability indices. In case the stability indices $\alpha_\xi$ and $\alpha_z$ are the same, the noise-corrupted time series $\{Y_t\}_{t \in \mathbb{Z}}$ follows the $S\alpha S$ distribution. Moreover, for values $\alpha_\xi = 2$ and $\alpha_z = 2$, \Cref{eq:noise-corrupted AR model} transforms into Gaussian autoregressive model with Gaussian additive noise, thus representing a noise-corrupted model with finite variance.
\subsection{Dependence measures}
We note that only the Gaussian distribution has finite variance from the class of $S\alpha S$ distribution (when $\alpha$ =2) rendering the classical autocovariance-based dependence measures ineffective for $S\alpha S$ distributions when $\alpha < 2$. Therefore, the notion of some alternative dependence measure must exist for infinite-variance processes. In light of this, FLOC \cite{XinyuMa1996, Nikias1994} is introduced as an alternative dependence measure. In the context of $S\alpha S$ random variables, FLOC is defined as follows:
\begin{definition}
Let $X$ and $Y$ be jointly $S\alpha S$ random variables. Then the FLOC is defined as 
\begin{equation}
 FLOC(X,Y,A,B) = \mathbb{E}[X^{\langle A \rangle}Y^{\langle B \rangle}],
 \end{equation}
 where $A,B \geq 0 $ satisfy the condition $A+B < \alpha$. The signed power $X^{\langle A \rangle}$ is given by $X^{\langle A \rangle} = |X|^Asign(X)$.
\end{definition}
FLOC can be used to measure the interdependence of a given stationary time series $\{X_t\}_{t \in \mathbb{Z}} $ with $S\alpha S$ distribution. In the said context, we define it as auto-FLOC, which is a function of the lag k between two time series data points $X_t$ and $X_{t-k}$.\\
We note that the modified YW method based on FLOC dependency measure (FLOC-YW) has been proposed in  \cite{Zulawinski2022} to estimate the $\alpha$-stable time series model. Furthermore, the FLOC-based EIV method \cite{Zulawinski2024} is employed to estimate the noise-corrupted AR model, where the noise or the AR model itself has infinite variance, thereby extending the EIV methodology from \cite{Diversi2007} for estimating an AR model in the presence of noise with finite variance.\\
The notion of FLOC of a random variable with itself generalizes the variance and is used to measure the dispersion of infinite-variance random variables \cite{Nikias1994, Zulawinski2025}.
\begin{remark}
Covariation is another notable dependence measure defined for jointly $S\alpha S$ random variables when $\alpha >1$. For more information, we refer to the bibliographic position \cite{samorodnitsky1994stable}. Covariation-based modified YW method \cite{Kruczek2017} has also been proposed for estimation of $\alpha$-stable time series models.
\end{remark}
\begin{remark}
For two independent $S\alpha S$ random variables, we have $FLOC(X,Y,A,B) = 0$.
\end{remark}
\begin{remark}
FLOC reduces to classical covariance for a zero-mean process when $\alpha = 2$ and $A = B =1$.
\end{remark}

\section{Denoising methodology}\label{sec:method}
In this section, we introduce our method to recover the pure $S\alpha S$-AR model from its additive noise-corrupted version. Our method is strongly motivated by Noise2Noise \cite{lehtinen2018}. First, we recall the Noise2Noise method for image denoising and then we delve into the proposed methodology later in the section.
\subsection{Background-Noise2Noise}
In problems related to image denoising, supervised learning typically requires a pair of noisy image $y$ and clean image $x$. The denoising neural network therefore trains by mapping the noisy image to the clean version of it. This approach has two major backlogs. Firstly, such pairs of noisy and clean data may not be available. Secondly, if there are significant differences between noise statistics in noisy and clean images, the denoising network may suffer from a domain-gap problem. To address these issues, \cite{lehtinen2018} proposed a learning method in which the network learns by mapping one noisy version $y_1 = x+e_1$ of the clean image $x$ to another noisy version $y_2 = x+e_2$ of it, $e_1,e_2$ being independent, zero-mean noise. This learning approach is justified by the following result.
\begin{theorem}
    Let $y_1$ and $y_2$ be two noisy fixed observations of the same clean image x, i.e; $y_1 = x + e_1$ and $y_2 = x + e_2$ where $e_i$ is the noise. Given that $e_i$ are independent and $\mathbb{E}[e_i] =0$, it follows that
    \begin{equation}
     \underset{\phi}{arg min}\, \mathbb{E}[ \|\Phi(y_1;\phi) - y_2\|^2_2] =  \underset{\phi}{arg min}\, \mathbb{E}[ \|\Phi(y_1;\phi) - x\|^2_2],
     \end{equation}
    where $\Phi$ is the denoising network with parameters $\phi$.
\end{theorem}
Therefore, training a network with respect to the mean squared error (MSE) with two noisy instances of the same image is equivalent to training with pairs of clean and noisy images. This approach allows supervised learning without needing the ground-truth images. In practice, if the size of the training data is not too small, Noise2Noise performs almost on par with the learning performed by mapping a noisy image to a clean image.
\subsection{Stable-N2N}\label{subsec:stable-n2n}

The proposed method `Stable-N2N' extends the Noise2Noise methodology in the context of denoising additive 
noise-corrupted $S\alpha S$-AR time series model. In such a model, as noted by \cite{Zulawinski2024}, the existence of a unique solution of $\{X_t\}_{t \in \mathbb{Z}}$ in the form of \Cref{eq:causal} is presumed. Consequently, the solution becomes stationary. The existence of such a solution is an intricate assumption for the estimation of pure time series. The goal of our approach is, given the noise-corrupted time series $\{Y_t\}_{t \in \mathbb{Z}}$ in \Cref{eq:noise-corrupted AR model}, to 
learn the stationary properties of the underlying time series $\{X_t\}_{t \in \mathbb{Z}}$. In particular, when a time series is second-order stationary, stochastic properties such as the mean and interdependence of random variables in the time series do not depend on $t$. We aim to capture these properties of the pure time series through self-supervised learning by mapping one random vector from the noise-corrupted process $\{Y_t\}_{t \in \mathbb{Z}}$ to another random vector from the same process. Our learning approach thus avoids the need to have complete knowledge about the noise model or access to data points from $\{X_t\}_{t \in \mathbb{Z}}$. We justify our approach by the following proposition.
\begin{prop}\label[prop]{res:stable n2n}
We consider a time series $\{Y_t\}_{t \in \mathbb{Z}}$ arising from the model governed by \Cref{eq:noise-corrupted AR model} when $(\alpha_\xi, \alpha_z) \ne (2, 2)$.
Then for random vectors of $\{Y_t\}_{t \in \mathbb{Z}}$, namely $\{Y_{s+1},\ldots,Y_{s+q}\}$ and\\
$\{Y_{s+q+N+1},\ldots,Y_{s+2q+N}\}$, where $N \in \mathbb{N} \cup \{0\}$,
it follows that, 
\begin{equation}\label{eq:n2n for time series with infinite variance}
 \underset{\phi}{arg min}\, \mathbb{E}[ \|\Phi(Y_{s+1}^{\langle B'\rangle},\ldots,Y_{s+q}^{\langle B'\rangle};\phi) -(Y_{s+q+N+1},\ldots,Y_{s+2q+N}) \|_2^2] 
 \end{equation}
 \begin{equation}
 =\underset{\phi}{arg min}\, \mathbb{E}[ \|\Phi(Y_{s+1}^{\langle B'\rangle},\ldots,Y_{s+q}^{\langle B'\rangle};\phi) - (X_{s+q+N+1},\ldots,X_{s+2q+N})\|_2^2],
\end{equation}
where $\Phi$ is the denoising neural network with parameters $\phi$,
$0< B'<\frac{min(\alpha_\xi,\alpha_z)}{2}$.
\end{prop}
\begin{proof}
Let 
$\Phi(Y_{s+1}^{\langle B'\rangle},\ldots,Y_{s+q}^{\langle B'\rangle};\phi) = (\overline{Y_{s+1}},\ldots,\overline{Y_{s+q}}) $.
\begin{align*}
&\underset{\phi}{\mbox{argmin}}\,\mathbb{E}[ \|\Phi(Y_{s+1}^{\langle B'\rangle},\ldots,Y_{s+q}^{\langle B'\rangle};\phi) - (Y_{s+q+N+1},\ldots,Y_{s+2q+N})\|_2^2]\\
&=\underset{\phi}{\mbox{argmin}}\, \mathbb{E}[ \|\Phi( Y_{s+1}^{\langle B'\rangle},\ldots,Y_{s+q}^{\langle B'\rangle};\phi)\|_2^2 - 2 \sum_{i=1}^qY_{s+q+N+i} \overline{Y_{s+i}}]\\ &\text{(expanding inside the euclidean norm)}\\
&= \underset{\phi}{\mbox{argmin}}\,\mathbb{E}[ \|\Phi( Y_{s+1}^{\langle B'\rangle},\ldots,Y_{s+q}^{\langle B'\rangle};\phi)\|_2^2 - 2 \sum_{i=1}^q(X_{s+q+N+i} + Z_{s+q+N+i}) \overline{Y_{s+i}}].\\
\end{align*}
Now let us note that as per the assumptions made in the context of \Cref{eq:noise-corrupted AR model} in \Cref{subsec:noisy ar model}, $\{Z_t\}_{t \in \mathbb{Z}}$ is i.i.d. $S\alpha S$ sequence independent of $\{X_t\}_{t \in \mathbb{Z}}$. Therefore $Z_{s+q+N+i}$ is independent of the random vector $(Y_{s+1},\ldots,Y_{s+q})$ and also $E[Z_t] = 0$.\\
Therefore,
\begin{align*}
&\underset{\phi}{\mbox{argmin}}\,\mathbb{E}[ \|\Phi( Y_{s+1}^{\langle B'\rangle},\ldots,Y_{s+q}^{\langle B'\rangle};\phi)\|_2^2 - 2 \sum_{i=1}^q(X_{s+q+N+i} + Z_{s+q+N+i}) \overline{Y_{s+i}}]\\
&=\underset{\phi}{\mbox{argmin}}\, \mathbb{E}[ \|\Phi(  Y_{s+1}^{\langle B'\rangle},\ldots,Y_{s+q}^{\langle B'\rangle};\phi)\|_2^2 - 2 \sum_{i=1}^q(X_{s+q+N+i})\overline{Y_{s+i}}]\\
&= \underset{\phi}{\mbox{argmin}}\, \mathbb{E}[ \|\Phi( Y_{s+1}^{\langle B'\rangle},\ldots,Y_{s+q}^{\langle B'\rangle};\phi) - (X_{s+q+N+1},\ldots,X_{s+2q+N})\|_2^2].
\end{align*}
\end{proof}
We consider $(Y_{s+1}^{\langle  B'\rangle},\ldots,Y_{s+q}^{\langle B'\rangle})$ as an input to our denoising network from a computational point of view, since $\mathbb{E}[Y_t^2]$ does not exist when $\alpha_\xi$ and $\alpha_z$ are not simultaneously equal to 2 (following the notations in \Cref{subsec:noisy ar model}). The motivation comes from the notion of generalized variance (see \cite{Nikias1994}). Therefore in the context of denoising noise-corrupted AR model with Gaussian distribution, we take the following approach.
\begin{corollary}\label{res:gaussian n2n}
For a time series $\{Y_t\}_{t \in \mathbb{Z}}$ arising out of the model governed by \Cref{eq:noise-corrupted AR model} when $(\alpha_\xi, \alpha_z) = (2, 2)$  and for random vectors of $\{Y_t\}_{t \in \mathbb{Z}}$, namely $\{ Y_{s+1},\ldots,Y_{s+q}\}$ and $\{Y_{s+q+N+1},\ldots,Y_{s+2q+N}\}$, where $N \in \mathbb{N} \cup \{0\}$,
it follows that 
\begin{equation}\label{eq:n2n for time series with finite variance}
 \underset{\phi}{\mbox{argmin}}\, \mathbb{E}[ \|\Phi(Y_{s+1},\ldots,Y_{s+q};\phi) -(Y_{s+q+N+1},\ldots,Y_{s+2q+N}) \|_2^2] 
 \end{equation}
 \begin{equation}
 =\underset{\phi}{\mbox{argmin}}\, \mathbb{E}[ \|\Phi(Y_{s+1},\ldots,Y_{s+q};\phi) - (X_{s+q+N+1},\ldots,X_{s+2q+N})\|_2^2],
\end{equation}
where $\Phi$ is the denoising neural network with parameters $\phi$.
\end{corollary}
We note that, although in the simulations presented in \Cref{sec:simulation}, we have considered the $S\alpha S$ distribution, the results presented in \Cref{res:stable n2n} and the subsequent \Cref{res:gaussian n2n} hold for other distributions as long as the considered moments exist. In this sense, the presented methodology is universal. We further explore this notion in \ref{app:universality}.

\subsection{Proposed methodology for denoising AR model with additive noise}\label{subsec:methodology}
Now, in light of the results presented in \Cref{subsec:stable-n2n}, we build our self-supervised learning approach so that, given noise-corrupted time series data $\{Y_t\}_{t=1}^{n}$, our denoising network can successfully extract the stochastic properties of the pure $S\alpha S$-AR time series data $\{X_t\}_{t=1}^{n}$:
\begin{itemize}
\item Let us consider samples from a $S\alpha S$-AR time series of order $p$, $\{X_t\}_{t=1}^{n}$ with i.i.d. innovation sequence drawn from distribution $S(\alpha_\xi,\sigma_\xi)$, $\alpha_\xi > 1$ and with i.i.d. additive noise $\{Z_t\}_{t=1}^{n}$ drawn from $S(\alpha_z,\sigma_z)$, $\alpha_z > 1$. Let the noise-corrupted time series data $\{Y_t\}_{t=1}^{n}$ be given by
\begin{equation*}
Y_t = X_t + Z_t,
\end{equation*}
for $t = 1,2,\ldots,n$
where,
\begin{equation*}
 X_t = \theta_1X_{t-1} + \theta_2X_{t-2}+\ldots+\theta_pX_{t-p} +\xi_t.
\end{equation*}
\item To construct a learning method from \Cref{res:stable n2n}, we set $N = 0$ so that our training algorithm can learn the stochastic properties of the random vector from the pure time series, that is close to the random vector we aim to obtain from its noise-corrupted version, at the time of inference. The network output is modeled as a multivariate random variable according to the multi-input multi-output (MIMO) approach since, as noted by \cite{Bontempi2008, BenTaieb2010}, MIMO helps to preserve stochastic dependencies in the time series. Assuming that we are processing vectors with $q$ components at a time, the training method is given by constructing windows consisting of input-output pairs, sliding over the entire noise-corrupted dataset $\{Y_t\}_{t=1}^n$.
\begin{align*}
Y_{q+1},\ldots,Y_{2q} &= \Phi(Y_1^{\langle B'\rangle},\ldots,Y_q^{\langle B'\rangle};\phi)\\
Y_{q+2},\ldots,Y_{2q+1} &= \Phi(Y_2^{\langle B'\rangle},\ldots,Y_{q+1}^{\langle B'\rangle};\phi)\\
                     & \vdotswithin{ = }\notag \\
Y_{n-(q-1)},\ldots,Y_n &= \Phi(Y_{n-(2q-1)}^{\langle B'\rangle},\ldots,Y_{n-q}^{\langle B'\rangle};\phi),
\end{align*}
where $\Phi(;\phi)$ is the chosen denoising network and $B'$ is a sufficiently small value. For the noise-corrupted AR model with Gaussian distribution, we choose $B'=1$ (from \Cref{res:gaussian n2n}).\\
Our denoising network for implementing the Stable-N2N method is a FNN with two fully connected hidden layers. The input layer has $10$ nodes, the same as the output layer for processing random vectors with fixed number of components, and the neural network has $22$ nodes in each of the two hidden layers. The network has been designed in this way, keeping its approximation capabilities in mind (see \cite{Kidger2020}), and to model time series of moderate length. We use Rectified Linear Unit (ReLU) \cite{Agarap2018} as the nonlinear activation function in each layer except the output layer, where we use linear activation. We are motivated to use ReLU for its significance in time series modeling and forecasting problems, for example, convolutional layers with ReLU are employed for conditional time series forecasting in \cite{borovykh2017}, or fully connected layers coupled with ReLU are used to model signals from the $\alpha$-stable distribution in \cite{Sathe2023}. In this context, we prefer an architecture with fully connected layers instead of convolutional layers to avoid the effect of smoothing kernels when modeling highly impulsive time series. We believe that a fully connected FNN has the opportunity to learn the required stochastic and temporal dependencies of the underlying pure time series $\{X_t\}_{t=1}^n$ from large enough sliding windows, leveraging the MIMO approach. The self-supervised training procedure is illustrated in \Cref{fig: training paradigm}.

\item After the network is trained to solve the optimization problem in the form of \Cref{eq:n2n for time series with infinite variance} through an iterative process (the number of iterations known as epochs), the trained network is used to construct the denoised trajectories. Thus, by \Cref{res:stable n2n}, in theory we are able to capture the interdependence of the random variables in the pure time series data $\{X_t\}_{t=1}^{n}$ through the solution of \Cref{eq:n2n for time series with infinite variance}. Here we note that, since the network output is a random vector, we get overlapped denoised time series samples. To construct the entire denoised trajectory from $\{Y_t\}_{t=1}^{n}$, we take the first point from each denoised vector except the last, where we use the entire vector.
\begin{align*}
\mathbf{\tilde{X_1}},\ldots,\tilde{X_q} &= \Phi(Y_1^{\langle B'\rangle},\ldots, Y_{q}^{\langle B'\rangle};\phi^*)\\
\mathbf{\tilde{X_2}},\ldots,\tilde{X}_{q+1} &= \Phi(Y_2^{\langle B'\rangle},\ldots, Y_{q+1}^{\langle B'\rangle};\phi^*)\\
                               & \vdotswithin{ = }\notag \\
 \mathbf{\tilde{X}_{n-(q-1)}},\ldots, \mathbf{\tilde{X_{n}}} &= \Phi(Y_{n-(q-1)}^{\langle B'\rangle},\ldots, Y_{n}^{\langle B'\rangle};\phi^*),
 \end{align*}
 where $\Phi(;\phi^*)$ denote the trained network with $\phi^*$ being the solution of \Cref{eq:n2n for time series with infinite variance}. 
\item Thus, we get the denoised time series $\{\tilde{X_{1}},\ldots,\tilde{X_n}\}$ given the noisy realizations $\{Y_t\}_{t=1}^{n}$.
\end{itemize}
The entire workflow consisting of training and inference on $\{Y_t\}_{t=1}^{n}$ is visualized in \Cref{fig: stable-n2n illustration}.
\begin{figure}[htpb]
  \begin{center}
  \includegraphics[width = \textwidth]{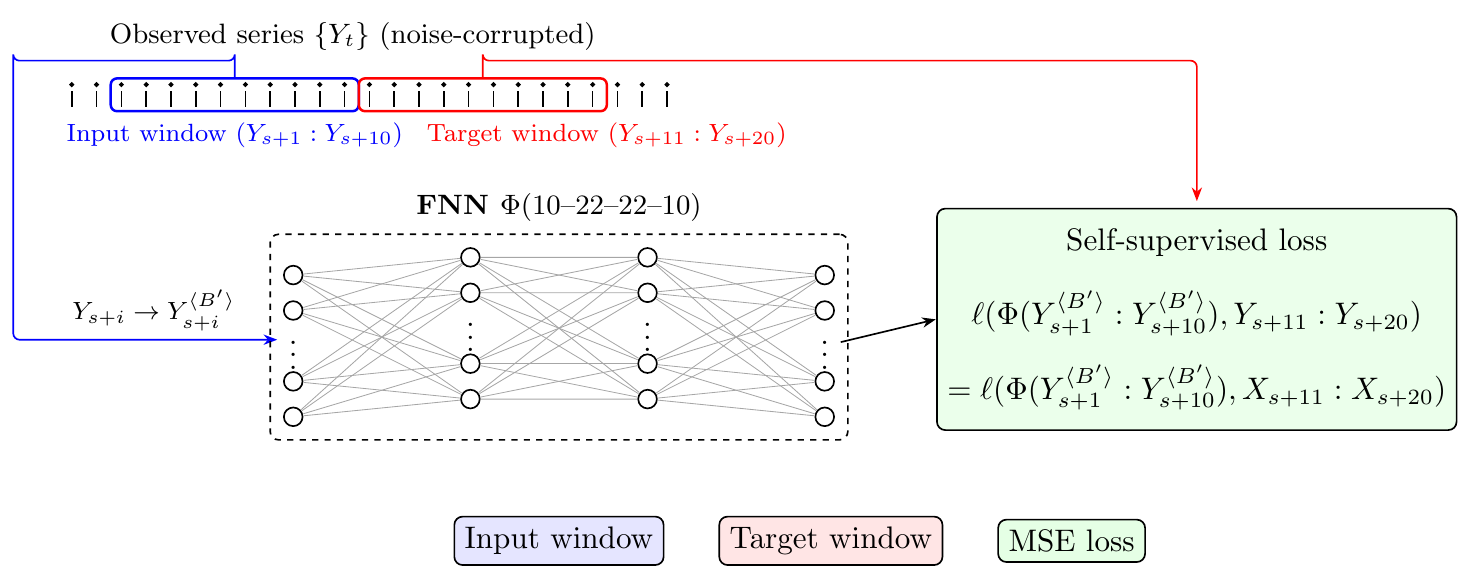}
    \caption {Schematics of training a fully connected FNN in Stable-N2N method.}
    \label{fig: training paradigm}
  \end{center}
\end{figure}

\begin{figure}[htpb]
  \begin{center}
  \includegraphics[width = \textwidth]{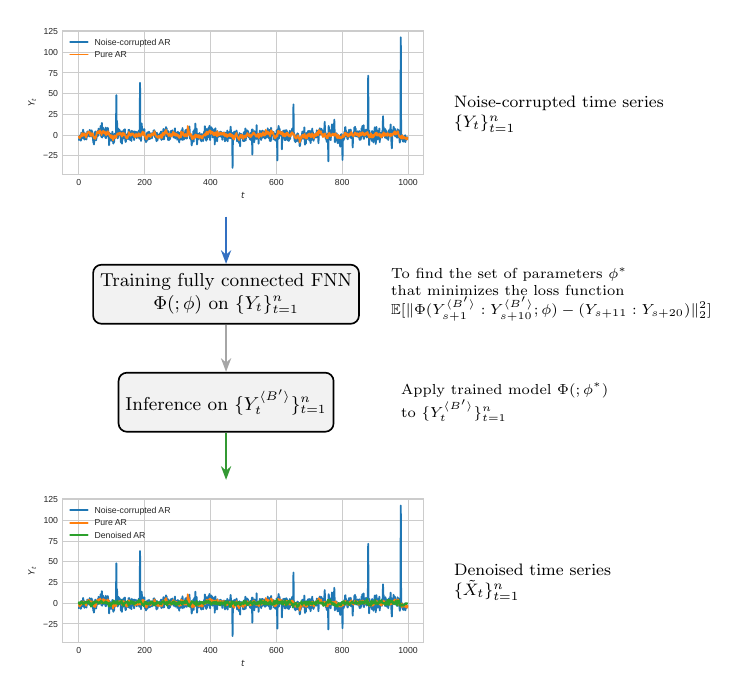}
    \caption {Workflow of the proposed Stable-N2N method.}
    \label{fig: stable-n2n illustration}
  \end{center}
\end{figure}

\subsection{Performance measures}\label{subsec: performance measures}
We validate our proposed methodology in the following way. From \Cref{eq:low order yw stable}, we note that the modified low-order YW (FLOC-YW) method (\Cref{eq:low order modified yw for pure ar}) leads to a biased estimation of the $S\alpha S$-AR($p$) model parameters if additional noise is present. We apply \Cref{eq:low order modified yw for pure ar} on the denoised data $\{\tilde X_t\}_{t=1}^{n}$ to estimate the parameters of the AR model given by $\Theta = [\theta_1,\ldots,\theta_p]$.
Following the notation in \ref{app:estimation infinite variance},
\begin{equation}\label{eq:modified yw on denoised data}
  \hat{\tilde\Gamma}^{\tilde x}\hat {\tilde \Theta} = \hat{\tilde\lambda}^{\tilde x}.
\end{equation}
If the pure data comes from the Gaussian AR model, we use the classical low-order YW method (\Cref{eq: classical low-order yw for pure ar}) on $\{\tilde X_t\}_{t=1}^{n}$ to recover the model parameters. If, through our proposed methodology, the additive noise present in the model is completely removed, the estimated parameters $\hat{\tilde\theta}_1,\ldots,\hat{\tilde\theta}_p$ should be unbiased and close to the true values of the model parameters $\theta_1,\ldots,\theta_p$. We compute the MAE between the true parameter values and the estimated parameters to quantify the bias associated with the estimation process. The formula is given by,
\begin{equation}\label{eq:mae denoising}
 \text{MAE} = \frac{1}{p} \sum_{i= 1}^{p} |\theta_i - \hat {\tilde {\theta_i}}|.
 \end{equation}
We rely on the estimation of the denoised data by the low-order YW equations to judge the denoising quality mainly for the following two reasons:
\begin{itemize}
\item \textbf{Estimation bias is related to noise strength:}\\
The term $\hat \Lambda$ in \Cref{eq:low order yw stable} is responsible for inducing bias in low-order YW equations, when parameters are estimated from noise-corrupted time series $\{Y_t\}_{t=1}^{n}$.
$\hat \Lambda$, empirical FLOC between $\{Y_t\}_{t=1}^{n}$ and $\{Z_t\}_{t=1}^{n}$ should increase with the noise strength because FLOC is a dependency measure and a stronger $\{Z_t\}_{t=1}^{n}$ influences $\{Y_t\}_{t=1}^{n}$ accordingly, since it is an additive noise in the model. In \ref{app:estimation infinite variance}, we empirically show that, indeed, $\hat \Lambda$ increases with stronger noise, thus, driving more bias into the YW equations. The behavior of low-order YW estimators under additive noise have already been noted in \cite{Diversi2007, Esfandiari2020, Zulawinski2023a}. Therefore, using low-order YW estimators on the denoised data acts as a measurement of relative noise strength. In \Cref{sec:simulation}, we will see that the MAE associated with the noise-corrupted data estimation using low-order YW equations increases with the noise strength, thus justifying our choice for the performance metric. If the noise-corrupted model is of finite variance, the variance of the additive noise $\nu_z$ is directly responsible for inducing bias, as seen in \Cref{eq:low order yw gaussian}.
\item \textbf{Estimation is intricately linked to the temporal structure of the underlying pure data}:
The AR model estimation heavily relies on its temporal structure. Therefore, given noise-corrupted data points, estimation from denoised data demonstrates the recovery of temporal relationships of the variables from the underlying pure time series, originally masked by the additive noise in the model. 
\end{itemize}

In addition, we use denoised samples to compute the $5$-steps ahead forecast of the pure time series $\{X_t\}_{t=1}^{n}$, assuming that the model order $p$ is known. If the noise-corrupted model is of finite variance, the order identification problem has already been addressed in \cite{Diversi2007, Zulawinski2023}. For the infinite-variance case, the correct order can be obtained by analyzing the numerical stabilization of the minimization problem \Cref{eq:loss function stable}, as mentioned in \cite{Diversi2007}, or by classical methods from the denoised data. We leave the detailed investigation of this aspect for future work.\\
For $p=2$, $5$-steps ahead forecast is given by, 
\begin{align*}
 \hat{X}_{n+1} &= \hat{\bar\theta}_1 \tilde{X_n} + \hat{\bar\theta}_2 \tilde X_{n-1}\\
  \hat{X}_{n+2} &= \hat{\bar\theta}_1\hat{X}_{n+1} +\hat{\bar\theta}_2 \tilde X_{n}\\
 \hat{X}_{n+i} &= \hat{\bar\theta}_1\hat{X}_{n+i-1} + \hat{\bar\theta}_2\hat{X}_{n+i-2}, \quad \mbox{for}\quad i = 3, 4, 5,
  \end{align*}
where $\hat {\bar\theta}_1, \hat {\bar\theta}_2$ are the parameters estimated from $\{Y_t\}_{t=1}^{n}$ by EIV method (\Cref{eq:eiv for stable}). Note that, when both the noise and the pure AR model are Gaussian, we use \Cref{eq:eiv gaussian} to estimate the parameters of the noise-corrupted model. MAE between the forecast and the original data points is computed to judge the accuracy of the forecast.\\
We choose to use the parameters estimated from $\{Y_t\}_{t=1}^{n}$ to construct $5$-steps ahead forecast to maintain uniformity in parameter estimation for all the considered baseline methods. Our goal is to show that, despite the availability of the EIV method for the estimation of a noise-corrupted model, denoising is crucial for forecasting the pure time series from its noise-corrupted version. In this context, forecasting of the pure time series serves as an important application of the proposed denoising methodology.
\subsection{Baselines and upper bound} 
We propose the following baselines and an upper performance bound to demonstrate the efficacy of the proposed method.
\begin{itemize}
\item Performance baselines: 
\begin{itemize}
\item \textbf{NAC}:  
NAC is a self-supervised learning method originally implemented in the context of image denoising in \cite{Xu2020}. Given noise-corrupted samples $\{Y_t\}_{t=1}^{n}$,  we add simulated noise $\{\breve Z_t\}_{t=1}^{n}$, statistically close to additive noise $\{Z_t\}_{t=1}^{n}$, to $\{Y_t\}_{t=1}^{n}$ and generate noisier samples $\{\breve Y_t\}_{t=1}^{n}$. The training methodology looks like
\begin{equation}\label{eq:nac}
Y_{t-(q-1)},\ldots, Y_t = \Phi(\breve Y_{t-(q-1)},\ldots, \breve Y_t;\phi), \quad t \geq q,
\end{equation}
where $\Phi(;\phi)$ is the denoising network with parameters $\phi$.
After training the network, the inference methodology is given by
\begin{equation}
\tilde X_{t-(q-1)},\ldots,\tilde X_t = \Phi(Y_{t-(q-1)},\ldots, Y_t;\phi^*), \quad t \geq q,
\end{equation}
where $\phi^*$ is the solution of the optimization problem connected to training of the network.
NAC has been shown to perform well when the additive noise is weak.
\item \textbf{NR2N}: NR2N method, introduced in \cite{Moran2020}, trains similarly as in NAC, but at the inference stage, a certain transformation is applied. Given noise-corrupted samples $\{Y_t\}_{t = m_1}^{m_{n_0}}$ as a training set and $\{Y_t\}_{t=1}^{n}$ to denoise, the denoising network trains on noisier samples $\{\breve Y_t\}_{t=m_1}^{m_{n_0}}$ following \Cref{eq:nac} and inference is applied on $\{\breve Y_t\}_{t=1}^{n}$ to recover $\{X_t\}_{t=1}^{n}$ in the following way
\begin{equation}
\tilde X_{t-(q-1)},\ldots,\tilde X_t = 2\Phi(\breve Y_{t-(q-1)},\ldots, \breve Y_t;\phi^*)-(\breve Y_{t-(q-1)},\ldots, \breve Y_t), \quad t \geq q,
\end{equation}
assuming that $\{\breve Z_t\}_{t=1}^{n}$ and $\{\breve Z_t\}_{t=m_1}^{m_{n_0}}$ are statistically close to additive noise $\{Z_t\}_{t=1}^{n}$.
\item \textbf{WDN}: Without denoising or WDN method refers to computing all the performance measures in \Cref{subsec: performance measures} on the noise-corrupted samples $\{Y_t\}_{t=1}^{n}$. We consider this baseline method to judge the denoising quality of Stable-N2N and other baseline methods.\\ In \Cref{eq:modified yw on denoised data}, in place of the denoised dataset $\{\tilde X_t\}_{t=1}^{n}$, we use $\{Y_t\}_{t=1}^{n}$ to compute the estimated parameters
\begin{equation}\label{eq:modified yw on noisy data}
  \hat{\tilde\Gamma}^y \hat {\tilde\Theta} = \hat{\tilde\lambda}^y,
\end{equation}
or, when the pure AR model is Gaussian (see \Cref{eq: classical low-order yw for pure ar}),
\begin{equation}\label{eq:classical yw on noisy data}
  \hat{\Gamma}^y \hat {\tilde\Theta} = \hat{\lambda}^y.
\end{equation}
Since no additive noise is removed, \Cref{eq:classical yw on noisy data,eq:modified yw on noisy data} lead to biased estimation of the model parameters. Thus if we compare the results of \Cref{eq:classical yw on noisy data,eq:modified yw on noisy data} with parameters estimated from the data, denoised through our proposed method, the biasness in the estimation should reduce in the case of Stable-N2N. In other words, the proposed methodology should always perform better than WDN.\\ 
$5$-steps ahead forecast of the pure time series is directly computed from $\{Y_t\}_{t=1}^{n}$ in the following way.
For $p=2$,
\begin{align*}
 \hat{X}_{n+1} &= \hat{\bar\theta}_1 Y_n + \hat{\bar\theta}_2 Y_{n-1}\\
  \hat{X}_{n+2} &= \hat{\bar\theta}_1\hat{X}_{n+1} +\hat{\bar\theta}_2  Y_{n}\\
 \hat{X}_{n+i} &= \hat{\bar\theta}_1\hat{X}_{n+i-1} + \hat{\bar\theta}_2\hat{X}_{n+i-2}, \quad \mbox{for}\quad i = 3, 4, 5.
  \end{align*}
\end{itemize}
\item Performance upper bound:
\begin{itemize}
\item \textbf{N2C}: Noise2Clean or N2C is the supervised learning method. The method assumes access to samples from the pure time series
$\{X_t\}_{t \in Z}$ while training. 
Training with samples $\{Y_t\}_{t = m_1}^{m_{n_0}}$ and $\{X_t\}_{t = m_1}^{m_{n_0}}$ from models $\{Y_t\}_{t \in Z}$ and $\{X_t\}_{t \in Z}$ respectively, is followed as
\begin{equation}
X_{t-(q-1)},\ldots,X_t = \Phi(Y_{t-(q-1)},\ldots,Y_t ;\phi), \quad t \geq m_1+q-1.
\end{equation}
The trained model is then inferred on the samples $\{Y_t\}_{t=1}^{n}$ from the same noise-corrupted model trajectory of $\{Y_t\}_{t \in Z}$ the network trained on, to get the denoised dataset $\{\tilde X_t\}_{t=1}^{n}$.\\ 
We note that, our N2C method has very little chance of suffering from a domain-gap problem since train and test data are from the same time series trajectory. Although such assumptions may not be compatible with real scenarios, the method is expected to perform better than other methods and serves as a perfect upper bound on model performance to our proposed methodology, since our method do not assume access to $\{X_t\}_{t \in Z}$ and requires only $\{Y_t\}_{t=1}^{n}$ to train and learn the patterns of the pure model. However, if the underlying pure data are highly impulsive, there may be significant difference between the train and test dataset, resulting in poor generalization. We will come back to this issue in the next section. 
\end{itemize}
\end{itemize}
The denoising networks of the performance baselines and upper bound are the same FNN architecture, taken in \Cref{subsec:methodology} so that the only difference between the baselines and the proposed method is the technique for achieving self-supervised learning. We note that the superior performance of the supervised learning method as an upper bound can indicate the adequacy of the considered FNN architecture in processing random vectors from the time series trajectories.\\
From a computational efficiency point of view, for a time series of length $n$, the Stable-N2N method, using a sliding window approach, generates $n-2q+1$ random vectors with $q$ components for training in a single epoch, while self-supervised methods considered as baselines generate $n-q+1$ vectors in the same context. However, it is important to mention that the proposed method considers an optimization problem different from baselines.
\section{Simulation study}\label{sec:simulation}
In this section, we validate the proposed methodology on synthetic and semi-synthetic datasets by Monte Carlo simulations. 
\subsection{Implementation details}\label{subsec:implementation}
The denoising networks are trained using the AdamW optimizer \cite{LoshchilovH19}, with a constant learning rate of $0.001$ and a weight-decay parameter of $0.0001$ to mitigate overfitting. The networks are trained for $30$ epochs with a batch size of $10$.
We note that the set of hyperparameters of our designed FNN network remains invariant for all the considered time series trajectories for both synthetic and semi-synthetic data. The reason is that we lack sufficient information about the noise or the autoregression from the noise-corrupted data, particularly when the AR model is corrupted with infinite-variance $S\alpha S$ noise. Following the same logic, we set the value of $B'$, used in the training of $S\alpha S$ trajectories in the proposed method, $0.45$ throughout, a sufficiently small value, considering the model assumptions $\alpha_\xi,\alpha_z >1$, except when the noise-corrupted trajectories generated are Gaussian distributed, we set $B' = 1$ in that case. The neural network for each simulated time series trajectory is built using Keras \cite{chollet2015keras} in TensorFlow \cite{tensorflow2015-whitepaper} and Python. To stabilize the training dynamics, we train the networks without shuffling the random vectors between epochs, thereby maintaining the temporal structure of the data throughout.\\
Given samples $\{Y_t\}_{t=1}^{n}$ from the noise-corrupted model, the proposed methodology Stable-N2N
and the baseline method NAC train on $\{Y_t\}_{t=1}^{n}$ and do not require any additional data to recover $\{X_t\}_{t=1}^{n}$. However, NR2N and N2C are implemented as dataset-based methods. To this end, we generate additional samples from the same trajectory, the samples $\{Y_t\}_{t=1}^{n}$ are taken from. NR2N requires only the noisy realizations, while N2C requires samples from both noise-corrupted and pure models. \\
The value of $\bar B$ in  \ref{app:estimation infinite variance} to implement the EIV method is taken as $0.45$ for all the considered cases, since we do not have enough information on additive noise $S(\alpha_z,\sigma_z)$ from the noise-corrupted data. The value of $r$ in \ref{app:finite variance estimation} and \ref{app:estimation infinite variance} to construct high-order YW equations is taken as $2$.

\subsection{Results on synthetic data}\label{subsec:synthetic datasets}
We consider a few different cases in our simulation study to accommodate both finite- and infinite-variance models. We simulate $M = 1000$ trajectories of $\{X_t\}_{t=1}^{n}$, pure $S\alpha S$-AR time series of order $p = 2$ with AR parameters $(\theta_1,\theta_2) =(0.5,0.3)$ and sample size $n_a = 1010$ corresponding to each innovation sequence $\{\xi_t\}_{t=1}^{n}$ considered in the simulation study. Furthermore, we add simulated noise from the model in \Cref{subsec:noisy ar model} to the pure time series trajectories to generate the $1000$ trajectories of $\{Y_t\}_{t=1}^{n}$, the noise-corrupted $S\alpha S$-AR time series with sample size $n = 999$, for each of the noise models considered.
We note that the additional data points in the simulated pure time series trajectories are used to compute $5$-steps ahead forecasting accuracies in terms of MAE. Moreover, when dataset-based methods are implemented, $999$ additional samples are generated from each simulated trajectory considered for training.

\subsubsection{Gaussian AR model with Gaussian and non-Gaussian noise}\label{subsubsec: Gaussian AR model with noise}
Firstly, we consider the Gaussian AR model with Gaussian additive noise, which is a finite-variance model. Following the notations in \Cref{subsec:noisy ar model}, to emulate a Gaussian AR model, the innovation sequence $\{\xi_t\}_{t=1}^{n}$ is taken from the Gaussian distribution with mean zero, $N(0,\nu_\xi)$, with variance $\nu_\xi = 1$. In addition, additive noise $\{Z_t\}_{t=1}^n$ is also taken from the zero-mean Gaussian distribution $N(0,\nu_z)$, with variance $\nu_z$ ranging from $\nu_z = 5$ to $\nu_z = 15$. The generated trajectories are depicted in \Cref{fig: noisy gaussian trajectories}. We note that with increasing level of $\nu_z$, the noise in the model becomes stronger.\\
\begin{figure}[htpb]
  \begin{center} 
\includegraphics[width = \textwidth]{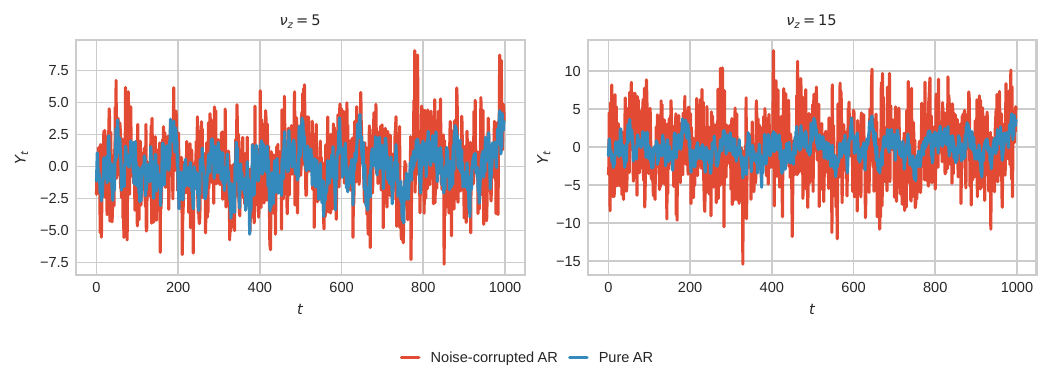}
  \caption {Sample trajectories of the Gaussian AR($2$) model with  Gaussian additive noise $N(0,\nu_z)$ with mean 0 and variance $\nu_z$.}
   \label{fig: noisy gaussian trajectories}
  \end{center}
\end{figure}
Further, we consider the Gaussian AR model with impulsive noise. To this end, the additive noise $\{Z_t\}_{t=1}^{n}$ is drawn from the distribution $S(\alpha_z,\sigma_z) $. Notably, the intensity of the additive noise depends on two parameters $\alpha_z$ and $\sigma_z$. The intensity of the noise increases with higher $\sigma_z$ whereas lower $\alpha_z$ results in sharp impulses in the data. Therefore, low $\alpha_z$ in combination with high $\sigma_z$ results in strongly impulsive noise. In the simulation study, we take $\alpha_z = 1.5,1.7$ with $\sigma_z$ ranging from $1$ to $2$ for each $\alpha_z$, so that different noise level scenarios can be analyzed. The sample trajectories are illustrated in \Cref{fig: gaussian with stable noise}. We note that the noise impulses are stronger for the sample trajectory in the right panel of the figure, being generated from relatively lower $\alpha_z$ and higher $\sigma_z$ as noise parameters.
\begin{figure}[htpb]
  \begin{center}
   \includegraphics[width = \textwidth]{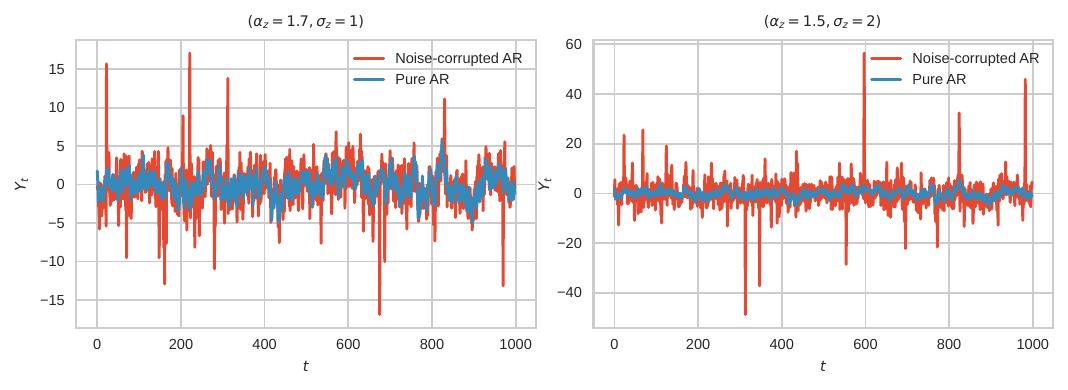}
    \caption{Sample trajectories of Gaussian AR($2$) model with additive noise from $S(\alpha_z,\sigma_z)$ distribution.}
    \label{fig: gaussian with stable noise}
  \end{center}
\end{figure}

Next, we present the box plots of the parameters estimated on the denoised data $\{\tilde X_t\}_{t=1}^{n}$ by the classical low-order YW method, since for both Gaussian and non-Gaussian noise, the underlying pure AR model is taken to be Gaussian. We note that to generate additional noise $\{\breve Z_t\}_{t \in \mathbb{Z}}$ for creating noisier samples in the NAC and NR2N method, we compute $\hat \nu_z$, an estimate of variance of additive noise, from \Cref{eq:loss function gaussian}, when the noise-corrupted model is Gaussian distributed and draw samples from the distribution $N(0,\hat \nu_z)$ accordingly. However, we do not have such estimates of the noise parameters when the additive noise $\{Z_t\}_{t \in \mathbb{Z}}$ is $S(\alpha_z,\sigma_z)$ distributed. In such cases, we perform \textbf{blind denoising} by drawing from the distribution $S(\alpha_{\breve z},\sigma_{\breve z})$ with $\alpha_{\breve z} \in [1.5,1.9]$ and $\sigma_{\breve z} \in [1,2.5]$ chosen uniformly for each simulated trajectory across all parameter values.\\
The estimation results for the denoised Gaussian AR model with Gaussian additive noise are shown in \Cref{fig: coeff_gaussian} and \Cref{tab:classical yw on finite variance data}. Stable-N2N outperforms the other baseline methods, despite NAC and NR2N being implemented with complete knowledge of noise. A possible reason is that NAC and NR2N train on noisier samples, thus never see the temporal and stochastic relationships in the pure time series while training. 

 \begin{figure}[htpb]
  \begin{center}
   \includegraphics[width = \textwidth]{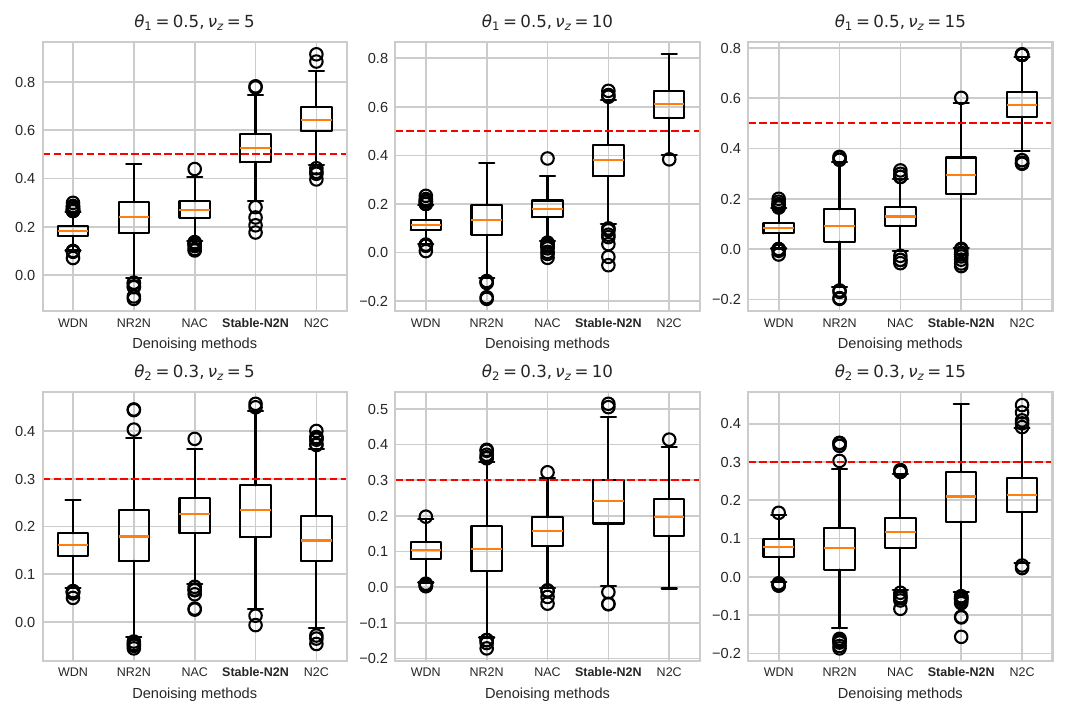}
    \caption {Box plots of parameters estimated with classical YW (\Cref{eq: classical low-order yw for pure ar}) on $1000$ denoised trajectories of Gaussian AR($2$) model with Gaussian additive noise $N(0,\nu_z)$.}
    \label{fig: coeff_gaussian}
  \end{center}
\end{figure}
\begin{table}[htpb]
\centering
\setlength{\tabcolsep}{12pt}
\begin{tabular}{l@{\hspace{1.5em}}ccccc}
\toprule
$\{Z_t\} \sim N(0,\nu_z)$ & WDN & NR2N & NAC & \textbf{Stable-N2N} & N2C \\
\midrule
$\nu_z$ = 5 & $0.2277$ & $0.1949$ & $0.1549$ & $\textbf{0.0783}$ & $0.1365$ \\
$\nu_z$ = 10 & $0.2902$ & $0.2807$ & $0.2331$ & $\textbf{0.1085}$ & $0.1117$ \\
$\nu_z$ = 15 & $0.3196$ & $0.3158$ & $0.2766$ & $\textbf{0.1618}$ & $0.0894$ \\
\bottomrule
\end{tabular}
\caption{Average MAE computed between true parameters and parameters estimated from classical YW (\Cref{eq: classical low-order yw for pure ar}) on $1000$ denoised trajectories of Gaussian AR($2$) model with Gaussian additive noise $N(0,\nu_z)$.}
\label{tab:classical yw on finite variance data}
\end{table}

\begin{figure}[htpb]
  \begin{center}
   \includegraphics[width = \textwidth]{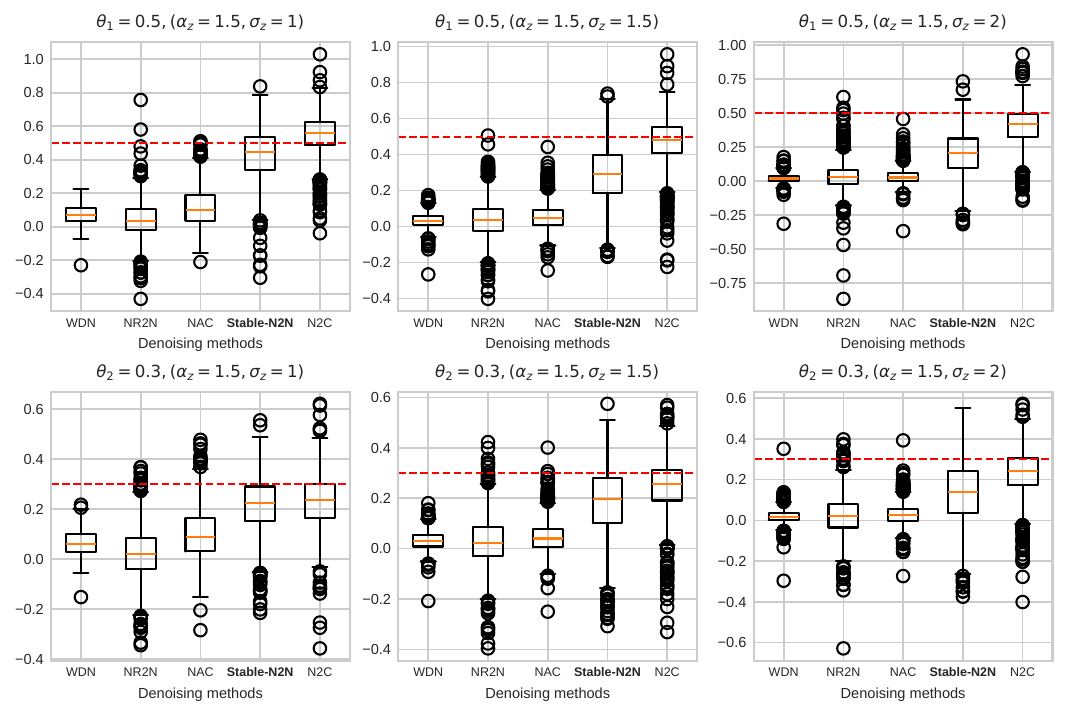}
    \caption {Box plots of parameters estimated with classical YW (\Cref{eq: classical low-order yw for pure ar}) on $1000$ denoised trajectories of Gaussian AR($2$) model with $S(\alpha_z = 1.5,\sigma_z)$ distributed additive noise.}
    \label{fig:coeff_gaussian_1.5}
  \end{center}
  \end{figure}
\begin{figure}[htpb]
  \begin{center}
   \includegraphics[width = \textwidth]{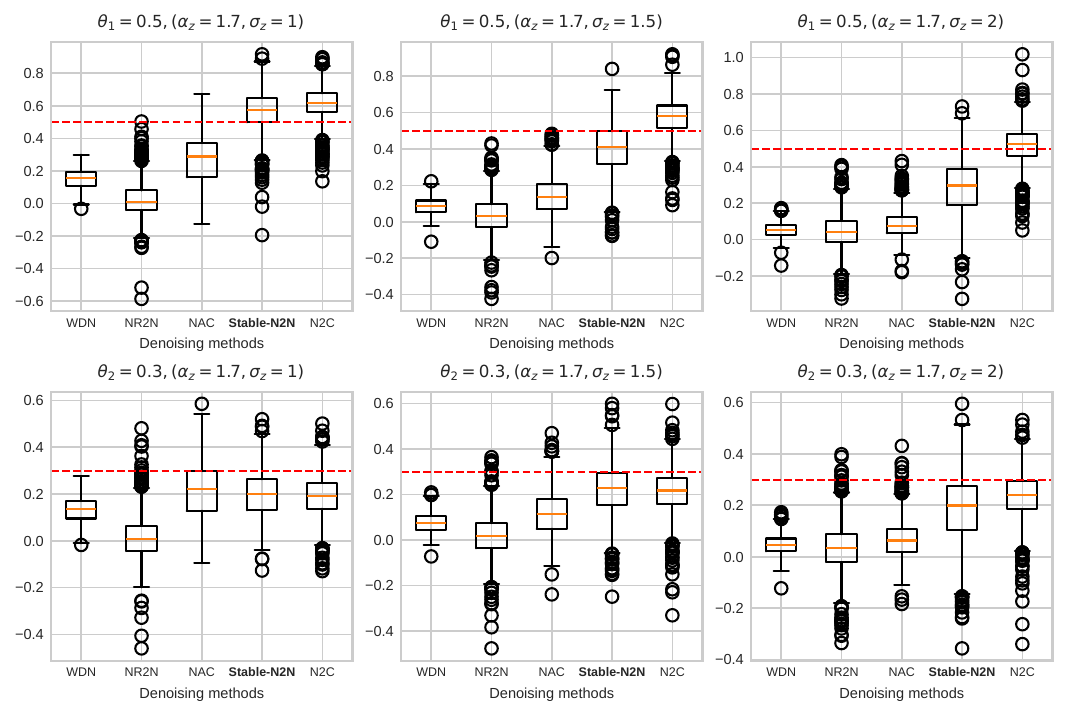}
    \caption {Box plots of parameters estimated with classical YW (\Cref{eq: classical low-order yw for pure ar}) on $1000$ denoised trajectories of Gaussian AR($2$) model with $S(\alpha_z = 1.7,\sigma_z)$ distributed additive noise.}
    \label{fig:coeff_gaussian_1.7}
  \end{center}
\end{figure}
\begin{table}[htpb]
\centering
\setlength{\tabcolsep}{12pt}
\begin{tabular}{l@{\hspace{2em}}ccccc}
\toprule
$\{Z_t\} \sim S(\alpha_z,\sigma_z)$ & WDN & NR2N & NAC & \textbf{Stable-N2N} & N2C \\
\midrule
$\alpha_z = 1.5,\sigma_z = 1$ & $0.3305$ & $0.3659$ & $0.2910$ & $\textbf{0.1209}$ & $0.1003$ \\
$\alpha_z = 1.5,\sigma_z = 1.5$ & $0.3648$ & $0.3678$ & $0.3484$ & $\textbf{0.1810}$ & $0.0938$ \\
$\alpha_z = 1.5,\sigma_z = 2$ & $0.3781$ & $0.3728$ & $0.3679$ & $\textbf{0.2413}$ & $0.1162$ \\
\midrule
$\alpha_z = 1.7,\sigma_z = 1$ & $0.2589$ & $0.3840$ & $0.1737$ & $\textbf{0.1163}$ & $0.1234$ \\
$\alpha_z = 1.7,\sigma_z = 1.5$ & $0.3193$ & $0.3722$ & $0.2687$ & $\textbf{0.1199}$ & $0.1025$ \\
$\alpha_z = 1.7,\sigma_z = 2$ & $0.3496$ & $0.3601$ & $0.3246$ & $\textbf{0.1793}$ & $0.0834$ \\
\bottomrule
\end{tabular}
\caption{Average MAE computed between true parameters and parameters estimated from classical YW (\Cref{eq: classical low-order yw for pure ar}) on $1000$ denoised trajectories of Gaussian AR($2$) model with $S(\alpha_z ,\sigma_z)$ distributed additive noise.}
\label{tab:classical yw on infinite variance data}
\end{table}
Estimation results for denoised Gaussian AR time series with non-Gaussian additive noise are presented in \Cref{tab:classical yw on infinite variance data} and illustrated in \Cref{fig:coeff_gaussian_1.5,fig:coeff_gaussian_1.7}. We note the bias in the estimation of the data denoised by NR2N and NAC, aggravated by the blind denoising. While the considered baseline methods find it harder to remove impulsive noise, Stable-N2N considerably outperforms them, even on par with the upper bound, N2C, for relatively weak noise. \\
Lastly, we compute the $5$-steps ahead forecast of the pure Gaussian AR model from the denoised data points. The average MAE between the forecast and the true values of the pure time series after $M = 1000$ simulations is given by
\begin{equation}\label{eq:average mae forecasting}
 E = \frac{1}{M}\sum_{i=1}^{M}\left(\frac{1}{5} \sum_{t= n+1}^{n+5} |X_t - \hat{X_t}|\right).
\end{equation}
The forecast results are depicted in \Cref{fig:gaussian_forecasting}.\\
We observe that Stable-N2N yields lower values of forecast accuracy ($E$) than the baseline methods, demonstrating the importance of denoising in the context of forecasting the pure AR model.

 \begin{figure}[htpb]
  \begin{center}
   \includegraphics[]{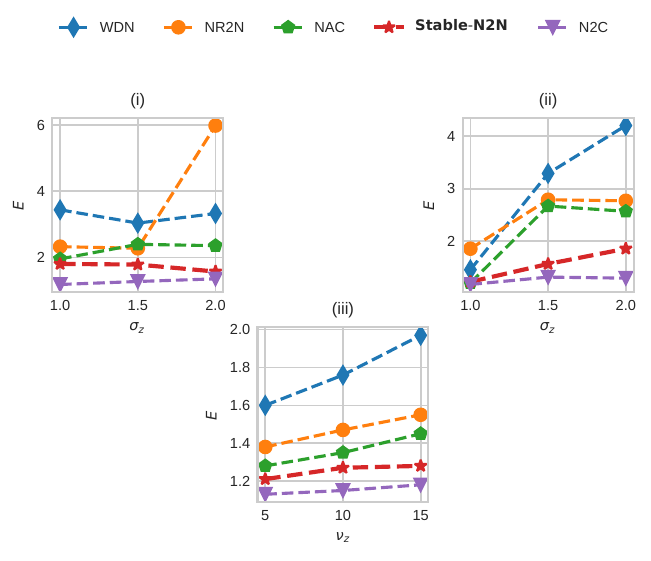}
    \caption {5-steps ahead forecast accuracy ($E$) computed from 1000 denoised trajectories of Gaussian AR($2$) model with (i) $S(\alpha_z = 1.5, \sigma_z)$, (ii) $S(\alpha_z = 1.7, \sigma_z)$, (iii) $N(0,\nu_z)$ distributed additive noise.}
    \label{fig:gaussian_forecasting}
  \end{center}
  \end{figure}
  
\subsubsection{Non-Gaussian AR model with non-Gaussian noise}\label{subsubsec:non-gaussian ar with non-gaussian noise}
In this section, we consider noise-corrupted trajectories, where the underlying pure AR model is impulsive. To this end, we generate AR model from the innovation sequence $\{\xi_t\}_{t=1}^{n}$ with the $S(\alpha_\xi, \sigma_\xi)$ distribution. We have previously noted that different values of the parameters in the distribution $S(\alpha, \sigma)$ result in different impulsive behaviors in the generated samples. Therefore, we consider samples from $S\alpha S$-AR time series trajectories generated by $\{\xi_t\}_{t=1}^{n} \sim S(\alpha_\xi = 1.9, \sigma_\xi = 1)$ and $\{\xi_t\}_{t=1}^{n} \sim S(\alpha_\xi = 1.5, \sigma_\xi = 0.5)$, respectively. Further, the non-Gaussian additive noise $\{Z_t\}_{t=1}^{n}$ is drawn from the distribution $S(\alpha_z,\sigma_z) $  with $\alpha_z = 1.5,1.7$ and $\sigma_z$ ranging from $1.5$ to $2.5$ corresponding to each $\alpha_z$. The sample trajectories are shown in \Cref{fig:stable with stable noise}. We observe that the pure $S\alpha S$-AR model trajectories generated from $\{\xi_t\}_{t=1}^{n} \sim S (\alpha_\xi = 1.5, \sigma_\xi = 0.5)$ and depicted in the lower panel of \Cref{fig:stable with stable noise}, showcase extreme fluctuations. The noise in the trajectories contains several impulses, coming from a non-Gaussian distribution with stronger noise generated from lower $\alpha_z$ and higher $\sigma_z$ as seen in the previous scenarios as well. Here, we note that the noise parameters in our analysis are chosen in such a way that the noise levels grow from weaker to stronger. Moreover, the weaker noise is still chosen to be significant, so that the proposed methodology can be differentiated from the baseline methods. Stronger noise levels are chosen in such a way that the autoregression remains significant compared to the additional noise in the model.
\begin{figure}[htpb]
  \begin{center}
   \includegraphics[width = \textwidth]{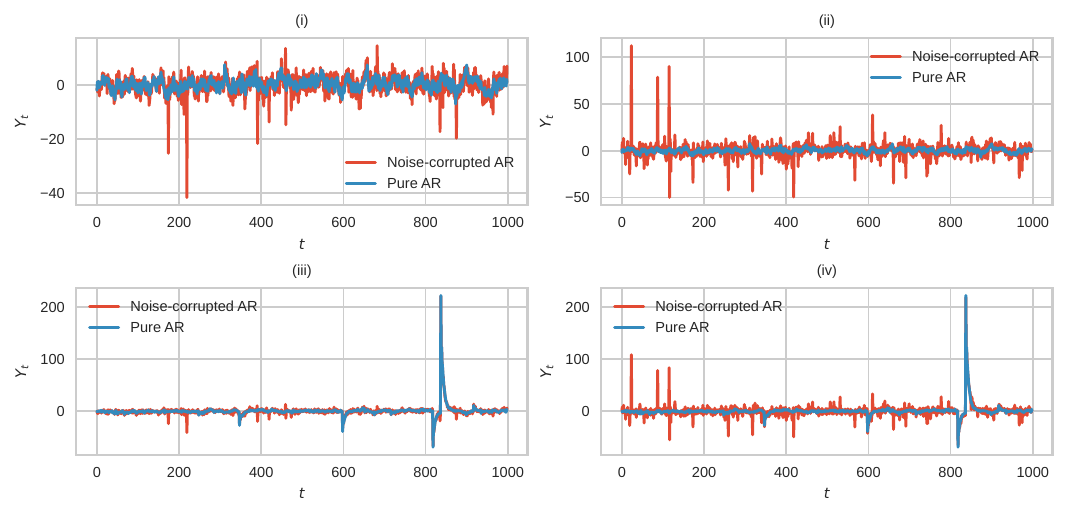}
     \caption {Sample trajectories of $S\alpha S$-AR($2$) model with additive $S\alpha S$ noise corresponding to (i) $ S(\alpha_\xi =1.9,\sigma_\xi =1)$, $S(\alpha_z = 1.7,\sigma_z = 1.5)$; (ii) $ S(\alpha_\xi =1.9,\sigma_\xi =1)$, $S(\alpha_z = 1.5,\sigma_z = 2.5)$; (iii) $ S(\alpha_\xi =1.5,\sigma_\xi = 0.5)$, $S(\alpha_z = 1.7,\sigma_z = 1.5)$; (iv) $ S(\alpha_\xi =1.5,\sigma_\xi = 0.5)$, $S(\alpha_z = 1.5,\sigma_z = 2.5)$. }
    \label{fig:stable with stable noise}
  \end{center}
\end{figure}
To create noisier samples in the implementation of NAC and NR2N, we follow the previously established setting in the case of blind denoising, that is, drawing from the distribution $S(\alpha_{\breve z},\sigma_{\breve z})$ blindly, with $\alpha_{\breve z} \in [1.5,1.9]$ and $\sigma_{\breve z} \in [1,2.5]$ chosen uniformly for each simulated trajectory across all parameter values.
Since the underlying pure model is infinite-variance, we use the modified YW method (\Cref{eq:low order modified yw for pure ar}) based on the notion of FLOC dependence measure on the denoised data to recover the model parameters. We use $A = 1, \: B = 0.45$ for both of the considered innovation sequences, satisfying the condition $A+B < \alpha_\xi$.\\
We present the estimation results from 1000 denoised trajectories of the noise-corrupted $S\alpha S$-AR($2$) model with innovation sequence $\{\xi_t\}_{t=1}^{n} \sim S(\alpha_\xi = 1.9, \sigma_\xi = 1)$ for each of the considered noise models in \Cref{tab:floc yw on infinite variance data} and illustrate them in \Cref{fig:coeff_stable_1.5,fig:coeff_stable_1.7}. Furthermore, estimation results for the considered innovation sequence $\{\xi_t\}_{t=1}^{n} \sim S(\alpha_\xi = 1.5, \sigma_\xi = 0.5)$ are presented in \Cref{tab:floc yw on infinite variance data(1)} and shown in \Cref{fig:coeff_stable(1)_1.5,fig:coeff_stable(1)_1.7}. We observe the presence of large outliers in \Cref{fig:coeff_stable(1)_1.5,fig:coeff_stable(1)_1.7} compared to \Cref{fig:coeff_stable_1.5,fig:coeff_stable_1.7}. This may be due to the extreme non-Gaussian nature of the pure model generated from $\{\xi_t\}_{t=1}^{n} \sim S(\alpha_\xi = 1.5, \sigma_\xi = 0.5)$. In particular, a dataset-based method like N2C, which otherwise performs better in most simulation settings, may face a significant difference between the train and the test dataset for some simulated trajectories. This issue is reiterated in the $5$-steps ahead forecast result shown in \Cref{fig:stable_forecasting}. We note the extreme values of the forecast error $E$ (\Cref{eq:average mae forecasting}) corresponding to dataset-based methods NR2N and N2C for some of the noise models. The results prove the superiority of the proposed methodology in recovering highly impulsive non-Gaussian AR data from its noise-corrupted versions. 
\begin{figure}[htpb]
  \begin{center}
   \includegraphics[width = \textwidth]{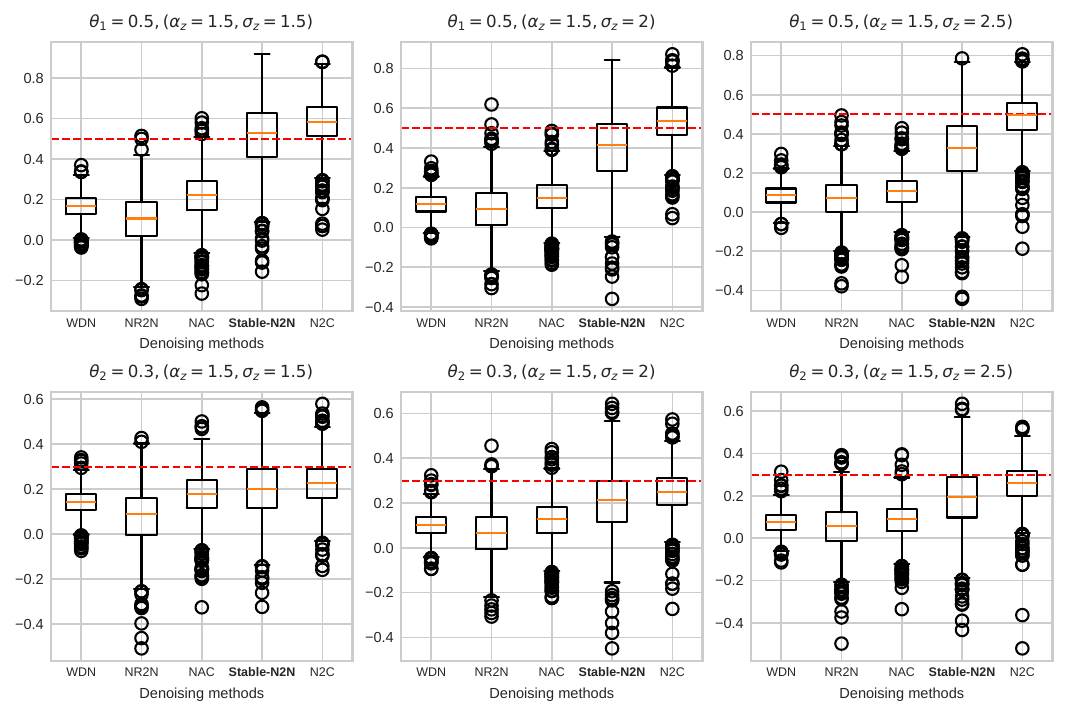}
    \caption {Box plots of parameters estimated with FLOC-YW (\Cref{eq:low order modified yw for pure ar}) on $1000$ denoised trajectories of noise-corrupted $S\alpha S$-AR($2$) model with innovation sequence $\{\xi_t\}_{t=1}^{n} \sim S(\alpha_\xi = 1.9, \sigma_\xi = 1)$ and $S(\alpha_z = 1.5,\sigma_z)$ distributed additive noise.}
    \label{fig:coeff_stable_1.5}
  \end{center}
  \end{figure}
\begin{figure}[htpb]
  \begin{center}
   \includegraphics[width = \textwidth]{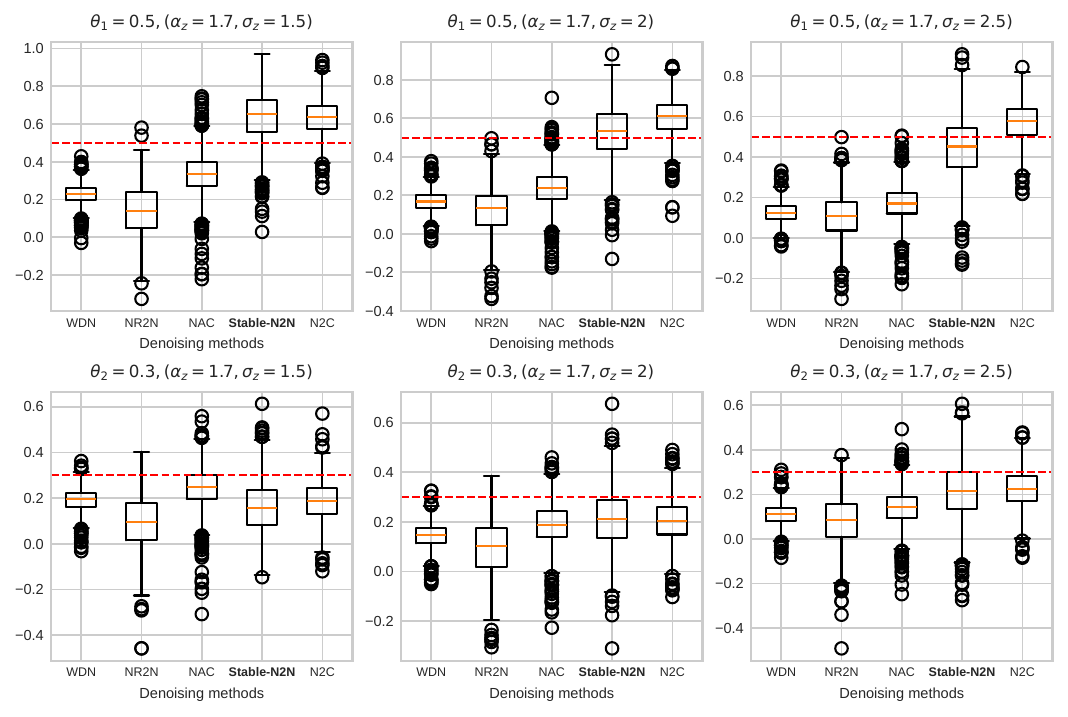}
    \caption {Box plots of parameters estimated with FLOC-YW (\Cref{eq:low order modified yw for pure ar}) on $1000$ denoised trajectories of noise-corrupted $S\alpha S$-AR($2$) model with innovation sequence $\{\xi_t\}_{t=1}^{n} \sim S(\alpha_\xi = 1.9, \sigma_\xi = 1)$ and $S(\alpha_z = 1.7,\sigma_z)$ distributed additive noise. }
    \label{fig:coeff_stable_1.7}
  \end{center}
\end{figure}
\begin{table}[htpb]
\centering
\setlength{\tabcolsep}{12pt}
\begin{tabular}{l@{\hspace{2em}}ccccc}
\toprule
$\{Z_t\} \sim S(\alpha_z,\sigma_z)$ & WDN & NR2N & NAC & \textbf{Stable-N2N} & N2C \\
\midrule
$\alpha_z = 1.5,\sigma_z = 1.5$ & $0.2479$ & $0.3123$ & $0.2108$ & $\textbf{0.1320}$ & $0.1066$ \\
$\alpha_z = 1.5,\sigma_z = 2$ & $0.2919$ & $0.3237$ & $0.2651$ & $\textbf{0.1441}$ & $0.0871$ \\
$\alpha_z = 1.5,\sigma_z = 2.5$ & $0.3206$ & $0.3397$ & $0.3054$ & $\textbf{0.1786}$ & $0.0850$ \\
\midrule
$\alpha_z = 1.7,\sigma_z = 1.5$ & $0.1917$ & $0.2870$ & $0.1275$ & $\textbf{0.1578}$ & $0.1321$ \\
$\alpha_z = 1.7,\sigma_z = 2$ & $0.2449$ & $0.2928$ & $0.1913$ & $\textbf{0.1156}$ & $0.1147$ \\
$\alpha_z = 1.7,\sigma_z = 2.5$ & $0.2830$ & $0.3082$ & $0.2462$ & $\textbf{0.1234}$ & $0.0950$ \\
\bottomrule
\end{tabular}
\caption{Average MAE computed between true parameters and parameters estimated from FLOC-YW (\Cref{eq:low order modified yw for pure ar}) on $1000$ denoised trajectories of noise-corrupted $S\alpha S$-AR($2$) model with innovation sequence $\{\xi_t\}_{t=1}^{n} \sim S(\alpha_\xi = 1.9, \sigma_\xi = 1)$ and $S(\alpha_z,\sigma_z)$ distributed additive noise. }
\label{tab:floc yw on infinite variance data}
\end{table}

\begin{figure}[htpb]
  \begin{center}
   \includegraphics[width = \textwidth]{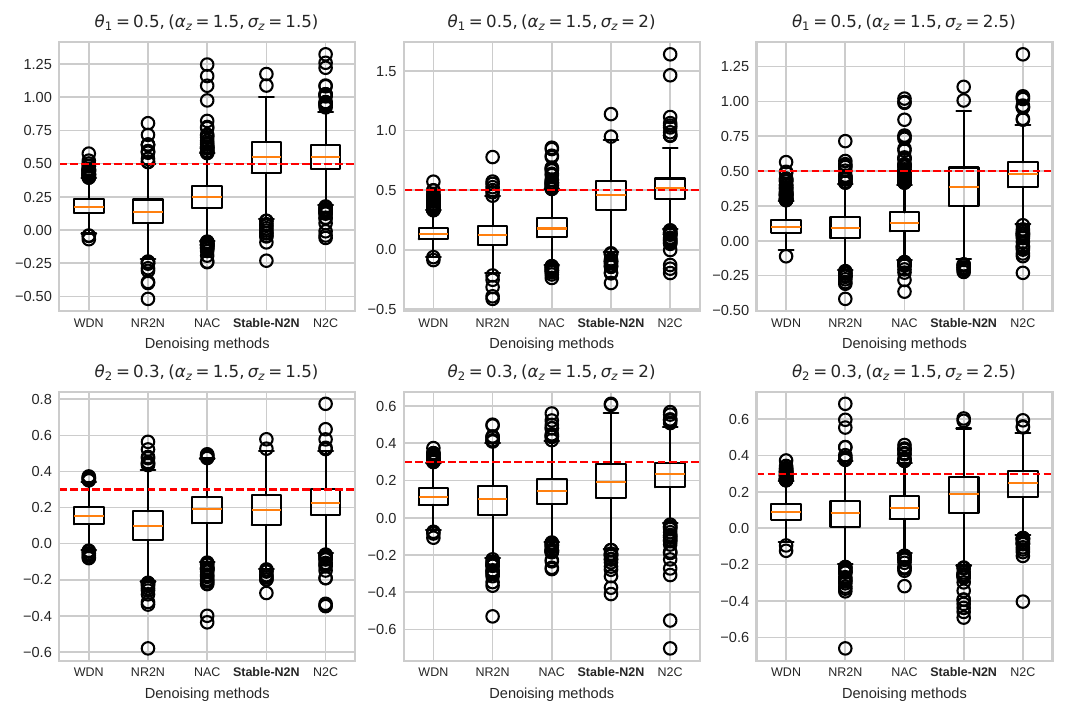}
    \caption {Box plots of parameters estimated with FLOC-YW (\Cref{eq:low order modified yw for pure ar}) on $1000$ denoised trajectories of noise-corrupted $S\alpha S$-AR($2$) model with innovation sequence $\{\xi_t\}_{t=1}^{n} \sim S(\alpha_\xi = 1.5, \sigma_\xi = 0.5)$ and $S(\alpha_z = 1.5,\sigma_z)$ distributed additive noise.}
    \label{fig:coeff_stable(1)_1.5}
  \end{center}
  \end{figure}
\begin{figure}[htpb]
  \begin{center}
   \includegraphics[width = \textwidth]{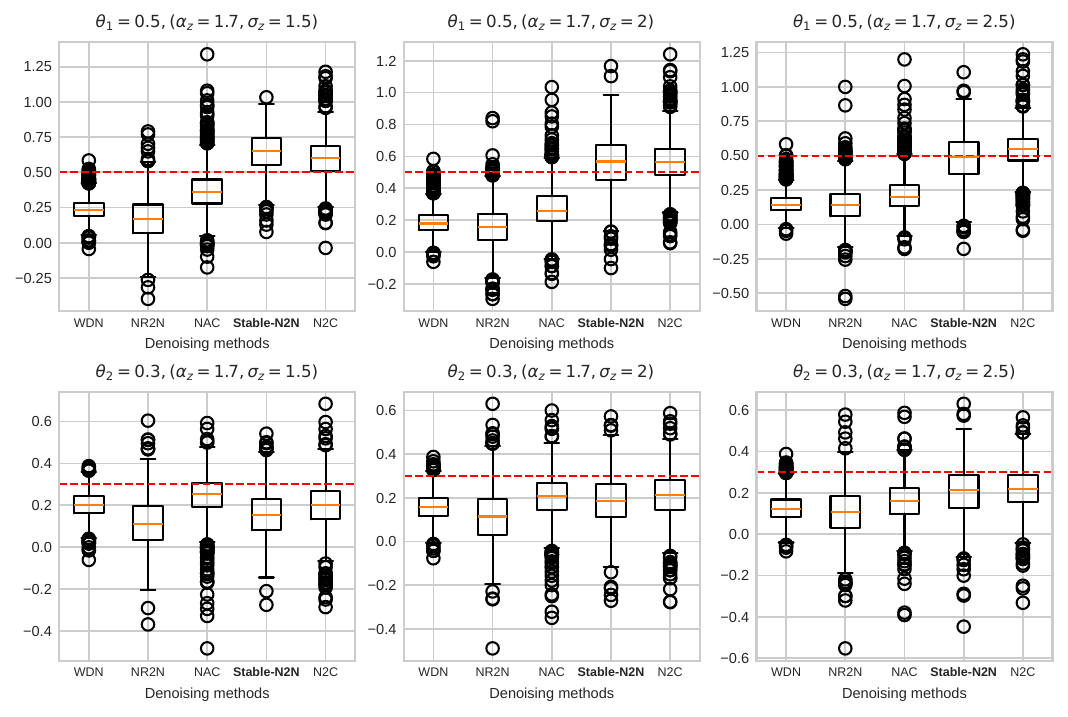}
    \caption {Box plots of parameters estimated with FLOC-YW (\Cref{eq:low order modified yw for pure ar}) on $1000$ denoised trajectories of noise-corrupted $S\alpha S$-AR($2$) model with innovation sequence $\{\xi_t\}_{t=1}^{n} \sim S(\alpha_\xi = 1.5, \sigma_\xi = 0.5)$ and $S(\alpha_z = 1.7,\sigma_z)$ distributed  additive noise. }
    \label{fig:coeff_stable(1)_1.7}
  \end{center}
\end{figure}
\begin{table}[htpb]
\centering
\setlength{\tabcolsep}{12pt}
\begin{tabular}{l@{\hspace{2em}}ccccc}
\toprule
$\{Z_t\} \sim S(\alpha_z,\sigma_z)$ & WDN & NR2N & NAC & \textbf{Stable-N2N} & N2C \\
\midrule
$\alpha_z = 1.5,\sigma_z = 1.5$ & $0.2304$ & $0.2850$ & $0.1968$ & $\textbf{0.1419}$ & $0.1171$ \\
$\alpha_z = 1.5,\sigma_z = 2$ & $0.2705$ & $0.2962$ & $0.2416$ & $\textbf{0.1470}$ & $0.1093$ \\
$\alpha_z = 1.5,\sigma_z = 2.5$ & $0.2979$ & $0.3143$ & $0.2777$ & $\textbf{0.1732}$ & $0.1076$ \\
\midrule
$\alpha_z = 1.7,\sigma_z = 1.5$ & $0.1797$ & $0.2643$ & $0.1266$ & $\textbf{0.1626}$ & $0.1315$ \\
$\alpha_z = 1.7,\sigma_z = 2$ & $0.2250$ & $0.2693$ & $0.1726$ & $\textbf{0.1352}$ & $0.1161$ \\
$\alpha_z = 1.7,\sigma_z = 2.5$ & $0.2583$ & $0.2804$ & $0.2175$ & $\textbf{0.1307}$ & $0.1084$ \\
\bottomrule
\end{tabular}
\caption{Average MAE computed between true parameters and parameters estimated from FLOC-YW (\Cref{eq:low order modified yw for pure ar}) on $1000$ denoised trajectories of noise-corrupted $S\alpha S$-AR($2$) model with innovation sequence $\{\xi_t\}_{t=1}^{n} \sim S(\alpha_\xi = 1.5, \sigma_\xi = 0.5)$ and $S(\alpha_z,\sigma_z)$ distributed additive noise. }
\label{tab:floc yw on infinite variance data(1)}
\end{table}

 \begin{figure}[htpb]
  \begin{center}
   \includegraphics[]{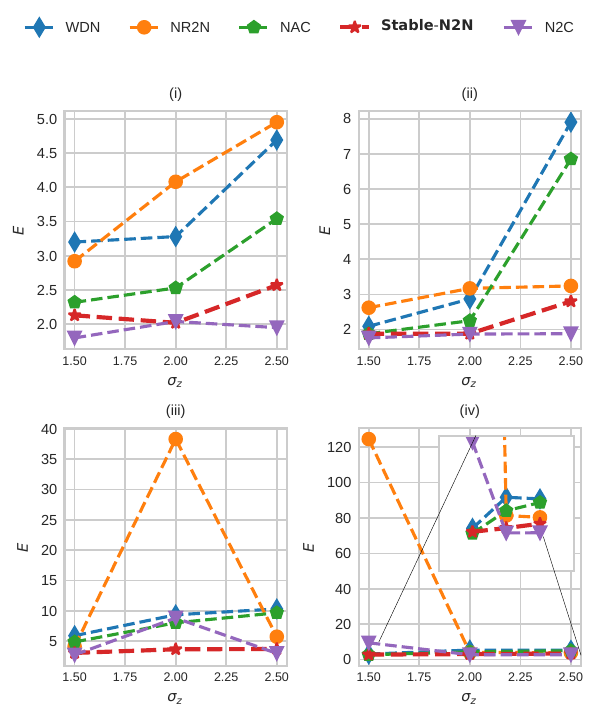}
    \caption {5-steps ahead forecast accuracy ($E$) computed from 1000 denoised trajectories of $S\alpha S$-AR($2$) model with additive $S\alpha S$ noise corresponding to (i) $ S(\alpha_\xi =1.9,\sigma_\xi =1)$, $S(\alpha_z = 1.5,\sigma_z)$; (ii) $ S(\alpha_\xi =1.9,\sigma_\xi =1)$, $S(\alpha_z = 1.7,\sigma_z)$; (iii) $ S(\alpha_\xi =1.5,\sigma_\xi = 0.5)$, $S(\alpha_z = 1.5,\sigma_z)$; (iv) $ S(\alpha_\xi =1.5,\sigma_\xi = 0.5)$, $S(\alpha_z = 1.7,\sigma_z)$. }
    \label{fig:stable_forecasting}
  \end{center}
  \end{figure}

\subsection{Results on semi-synthetic data} 
From \Cref{subsubsec:non-gaussian ar with non-gaussian noise} we observe that recovering the $S\alpha S$-AR model with extreme non-Gaussian nature from the noise-corrupted data is the most challenging among all the models considered. To further analyze this issue in a semi-real setting, we consider the USDPLN currency data from \cite{Sathe2023}. The dataset considered consists of $208$ observations taken from September 1, 2019 to July 1, 2020 and is shown in the upper left panel of \Cref{fig:real data}. We preprocess the dataset following the same way as in \cite{Sathe2023}. The lag-1 difference is taken to make the dataset stationary. The stationary dataset consists of $207$ observations and is shown in the upper right panel of \Cref{fig:real data}. Further, the transformed dataset is subdivided into a training set consisting of the first $167$ observations, and the remaining $40$ points are treated as the test dataset to compute the forecast accuracy. Due to the extreme fluctuation present in the data, the $S\alpha S$-AR(2) model is fitted to the training data in \cite{Sathe2023} and validated by analyzing its residuals. The residuals are found to follow $\alpha$-stable distribution with estimated parameters $(\hat \alpha,\hat \beta, \hat \sigma, \hat \mu) =(1.71,0.0008,0.003,0.009)$.\\
Now, we apply the modified YW method (\Cref{eq:low order modified yw for pure ar}) on the training data by taking $A = 1, \: B = 0.66$, the same values taken in \cite{Sathe2023} and estimate the model parameters. The estimated parameters are $(\hat \theta_1, \hat \theta_2) = (0.2177,0.1629)$. Moreover, we add simulated impulsive noise to the training dataset from the noise model $S(\alpha_z,\sigma_z)$ with $\alpha_z = 1.5,1.7$ and $\sigma_z$ ranging from $0.02$ to $0.06$ corresponding to each $\alpha_z$. We generate $1000$ noise-corrupted trajectories for each value of $(\alpha_z, \sigma_z)$ and apply our proposed methodology and the considered baseline methods on the generated semi-synthetic datasets to recover the USDPLN training data from its noise-corrupted versions. The trajectories of the noise-corrupted data are shown in the lower panel of \Cref{fig:real data}. The parameters of the simulated noise are chosen to create weakly to strongly impulsive noise in the data, keeping the autoregression significant in all the considered models.

\begin{figure}[htpb]
\begin{center}
\includegraphics[width = \textwidth]{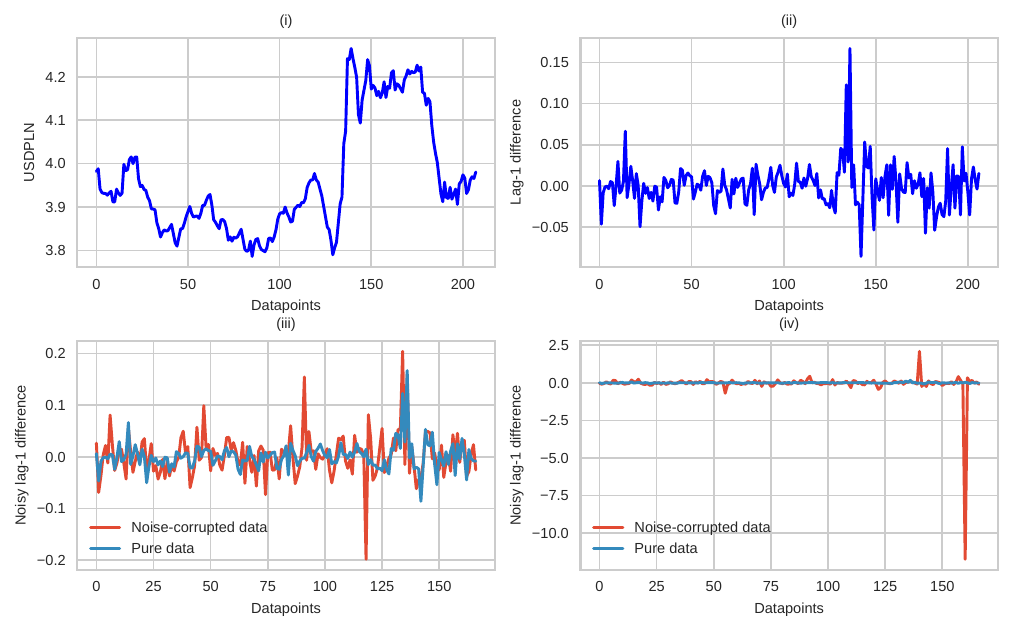}
\caption{(i) USDPLN currency data; (ii) Lag-1 difference; USDPLN stationary data with (iii) $S(\alpha_z = 1.7,\sigma_z = 0.02)$, (iv) $S(\alpha_z = 1.5,\sigma_z = 0.06)$ distributed additive noise.}
\label{fig:real data}
\end{center}
\end{figure}
Since NR2N and N2C are dataset-based methods and need additional data to train, we generate additional $167$ samples from each simulated noise-corrupted trajectory. Further, NAC and NR2N are implemented by creating noisier samples with additional simulated noise generated from
$S(\alpha_{\breve z},\sigma_{\breve z})$ distribution with $\alpha_{\breve z} \in [1.5,1.9]$ and $\sigma_{\breve z} \in [0.01,0.1]$ chosen uniformly for each simulated trajectory across all parameter values.\\
We apply the modified YW method (\Cref{eq:low order modified yw for pure ar}) on the denoised data to recover the model parameters by taking $A = 1, \: B = 0.66$. The model estimation results are presented in \Cref{tab:real coeff estimation} and displayed in \Cref{fig:coeff_real_1.5,fig:coeff_real_1.7}. We observe that our proposed methodology outperforms even the considered upper bound, N2C method, reaffirming our method's capability to recover highly impulsive data from its noise-corrupted versions.\\
The $5$-steps ahead forecast results are visualized in \Cref{fig:real_forecasting}.
\begin{figure}[htpb]
  \begin{center}
   \includegraphics[width = \textwidth]{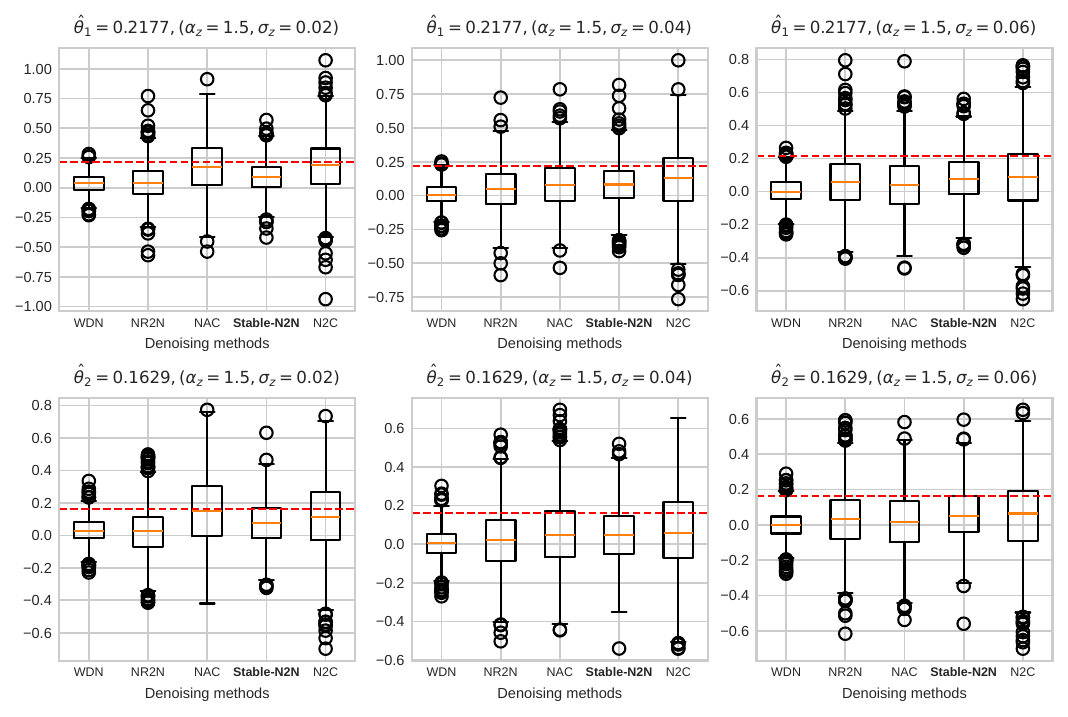}
    \caption {Box plots of parameters estimated with FLOC-YW (\Cref{eq:low order modified yw for pure ar}) on $1000$ denoised trajectories of noise-corrupted USDPLN data with $S(\alpha_z = 1.5,\sigma_z)$ distributed additive noise.}
    \label{fig:coeff_real_1.5}
  \end{center}
  \end{figure}
\begin{figure}[htpb]
  \begin{center}
   \includegraphics[width = \textwidth]{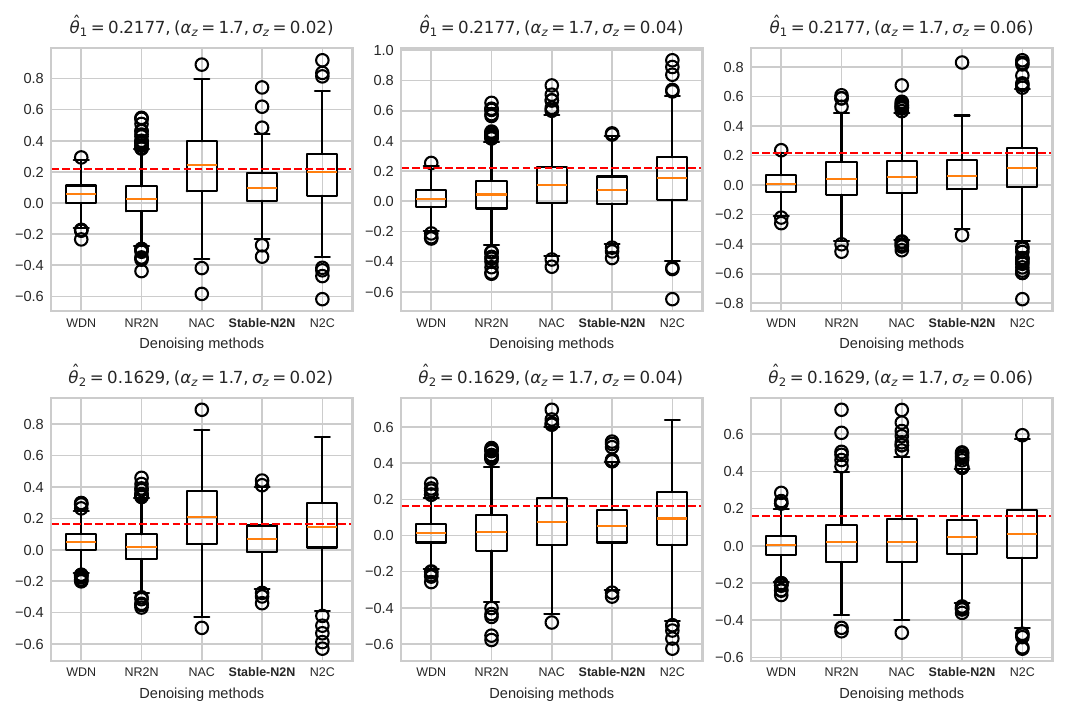}
    \caption {Box plots of parameters estimated with FLOC-YW (\Cref{eq:low order modified yw for pure ar}) on $1000$ denoised trajectories of noise-corrupted USDPLN data with $S(\alpha_z = 1.7,\sigma_z)$ distributed additive noise.}
    \label{fig:coeff_real_1.7}
  \end{center}
\end{figure}
\begin{table}[htpb]
\centering
\setlength{\tabcolsep}{12pt}
\begin{tabular}{l@{\hspace{2em}}ccccc}
\toprule
$\{Z_t\} \sim S(\alpha_z,\sigma_z)$ & WDN & NR2N & NAC & \textbf{Stable-N2N} & N2C \\
\midrule
$\alpha_z = 1.5,\sigma_z = 0.02$ & $0.1575$ & $0.1779$ & $0.1802$ & $\textbf{0.1392}$ & $0.1843$ \\
$\alpha_z = 1.5,\sigma_z = 0.04$ & $0.1826$ & $0.1848$ & $0.1819$ & $\textbf{0.1587}$ & $0.1945$ \\
$\alpha_z = 1.5,\sigma_z = 0.06$ & $0.1888$ & $0.1838$ & $0.1939$ & $\textbf{0.1567}$ & $0.1964$ \\
\midrule
$\alpha_z = 1.7,\sigma_z = 0.02$ & $0.1394$ & $0.1782$ & $0.1849$ & $\textbf{0.1335}$ & $0.1625$ \\
$\alpha_z = 1.7,\sigma_z = 0.04$ & $0.1753$ & $0.1831$ & $0.1714$ & $\textbf{0.1517}$ & $0.1764$ \\
$\alpha_z = 1.7,\sigma_z = 0.06$ & $0.1851$ & $0.1847$ & $0.1864$ & $\textbf{0.1611}$ & $0.1790$ \\
\bottomrule
\end{tabular}
\caption{Average MAE computed between parameters estimated from $1000$ denoised trajectories of noise-corrupted USDPLN data with $S(\alpha_z,\sigma_z)$ distributed additive noise and parameters estimated from pure USDPLN data, by FLOC-YW (\Cref{eq:low order modified yw for pure ar}) method.}
\label{tab:real coeff estimation}
\end{table}
 \begin{figure}[htpb]
  \begin{center}
   \includegraphics[]{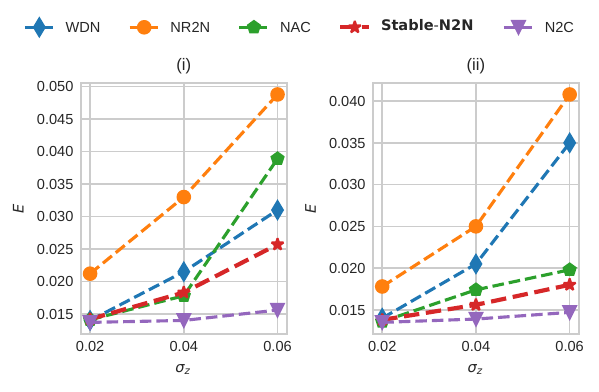}
    \caption {5-steps ahead forecast accuracy ($E$) computed from 1000 denoised trajectories of USDPLN data with (i) $S(\alpha_z = 1.5,\sigma_z)$, (ii) $S(\alpha_z = 1.7,\sigma_z)$ distributed additive noise. }
    \label{fig:real_forecasting}
  \end{center}
  \end{figure}
\section{Summary and discussion}\label{sec:conclusion}
In this paper, we present a novel self-supervised learning method in the context of time series denoising, to recover pure $S\alpha S$-AR model, $\alpha >1$, given samples from its additive noise-corrupted versions. The efficacy of our proposed learning method is demonstrated by comparing with other self-supervised methods in the literature on both finite- and infinite-variance models. In our simulation study, we consider a few different cases, with noise levels ranging from weakly impulsive to strongly impulsive. Moreover, the pure time series data are taken to be Gaussian as well as non-Gaussian in nature. For the non-Gaussian AR data, the simulation study considers data mildly non-Gaussian as well as extremely non-Gaussian in nature. We further pursue the case where the pure AR data contain large deviations by adding simulated impulsive noise to the USDPLN currency data and generating noise-corrupted samples, which are semi-synthetic in nature. The highlights of our method are that we do not need complete knowledge of the additive noise model and in addition, we do not require any additional data for the model to train. Therefore, our method performs considerably better than the baseline methods when the noise is strongly impulsive or when the pure time series data itself have extreme non-Gaussian nature. Also, it can be considered universal in the sense that the proposed methodology can be applied to other distribution models as long as the moments considered in our approach, exist.\\
% As noted in \cite{Zulawinski2023a}, the motivation behind considering noise-corrupted AR models is rooted in the real data analyzed in the literature. The article \cite{Wylomanska2014} noted the presence of additive noise in the data related to condition monitoring. Recovering the pure vibration signal from the corrupted version is crucial for local damage detection. Further, the bibliographic positions \cite{Sarnaglia2010, Solci2019} found the presence of additive outliers in the data corresponding to the concentration of particulate matter and developed robust estimation methods to accommodate large observations in the data.\\
In the presence of strongly impulsive noise, estimation and identification problems can be difficult. One may need the complete characteristics of the noise distribution. Furthermore, forecasting of pure data is inherently related to these problems. Therefore, denoising is a crucial aspect in heavy-tailed time series modeling and can have significant real applications. In this context, we believe that the proposed methodology will serve an important purpose. \\
In this study, we do not analyze the performance of denoising methods based on different deep learning architectures. Rather, our focus is on the training paradigms. Nonetheless, as mentioned earlier, the superiority of N2C as a performance upper bound shows that the considered FNN model is able to capture the dependencies in the time series data. In this regard, for a time series of very short length, the considered FNN architecture may need to be changed accordingly, to avoid over-parametrization.
Also, the proposed method, in conjunction with the FNN architecture, may face limitations when the underlying data are non-stationary. However, the assumption of stationarity remains standard in time-series modeling problems due to its significance in estimation and forecasting.
%%\newpag
% \section*{CRediT authorship contribution statement}
% \textbf{Sayantan Banerjee:} Conceptualization, Methodology, Software, Validation, Formal analysis, Investigation, Data Curation, Writing - original draft, Visualization.
% \textbf{Agnieszka Wy{\l}oma{\'n}ska:} Conceptualization, Investigation, Methodology, Resources, Supervision, Formal analysis, Writing - original draft, Writing - review \& editing.
% \textbf{S.Sundar:} Resources, Supervision, Methodology, Writing - review \& editing.

% \section*{Declaration of competing interest}
% The authors declare that they have no known competing financial interests or personal relationships that could have appeared to
% influence the work reported in this paper.
\section*{Acknowledgments}
The work of S.B. was supported by the Ministry of Human Resource Development (MHRD), Government of India, through institutional funding provided to the Indian Institute of Technology Madras.
The work of A.W. was supported by NCN OPUS project No. 2024/53/B/HS4/00433.
% \section*{Data availability}
% The authors do not have permission to share data.




% \end{thebibliography}
\appendix
\section{Estimation of noise-corrupted AR model with finite variance}\label{app:finite variance estimation}
In this part, we describe the EIV methodology proposed in \cite{Diversi2007} for estimating the parameters of a finite-variance AR model with i.i.d. additive noise of finite variance. The method combines low and high-order YW equations based on the classical autocovariance function $\gamma^y(k) = \mathbb{E}Y_tY_{t-k}$ defined for a zero-mean process.\\
First, we recall that for the pure autoregressive time series $\{X_t\}_{t \in \mathbb{Z}}$ with finite variance, low-order YW equations are defined as 
\begin{equation}\label{eq: classical low-order yw for pure ar}
    \Gamma^x\Theta = \lambda^x ,
\end{equation}
where, assuming that the AR model is of order $p$, $\Gamma^x$ is a non-singular matrix $p \times p$ given by 
\begin{equation}
(\Gamma^x)_{i,j} = \gamma^x(j-i) ,
\end{equation}
for $\gamma^x(k) = \mathbb{E}X_tX_{t-k}$, and 
\begin{equation}
\lambda^x = [\gamma^x(1),\ldots,\gamma^x(p)]', \quad \Theta = [\theta_1,\ldots,\theta_p]'.
\end{equation}
\Cref{eq: classical low-order yw for pure ar} is known as the classical YW method.\\
When $\{X_t\}_{t \in \mathbb{Z}}$ is corrupted by zero-mean additive noise of variance $\nu_z$, low-order YW equations are given by
\begin{equation}
(\Gamma^y-\nu_z I_p)\Theta = \lambda^y,
\end{equation}
where
\begin{equation}
(\Gamma^y)_{i,j} = \gamma^y(j-i),\quad \lambda^y = [\gamma^y(1),\ldots,\gamma^y(p)]',
\end{equation}
with $\Gamma^y$ being a $p \times p$ matrix.\\
High-order YW equations for the process $\{Y_t\}_{t \in \mathbb{Z}}$ are given by
\begin{equation}
\Gamma^y_r\Theta = \lambda_r^y,
\end{equation}
where $\Gamma^y_r$ is a matrix of order $r \times p$ given by
\begin{equation}
(\Gamma^y_r)_{i,j} = \gamma^y(p+i-j) \quad \mbox{and} \quad \lambda^y_r = [\gamma^y(p+1),\ldots,\gamma^y(p+r)]'.
\end{equation}
During estimation of the model, we consider the empirical versions of the matrices in high- and low-order YW equations. For samples $\{Y_1,\ldots,Y_n\}$ from noise-corrupted time series $\{Y_t\}_{t \in \mathbb{Z}}$, empirical version of $\gamma^y(k)$ is given by
\begin{equation}
\hat{\gamma}^y(k) = \frac{1}{l_2 - l_1}\sum_{t=l_1}^{l_2} Y_tY_{t-k},
\end{equation}
where
\begin{equation}
l_1 = max(1,1+k), \quad l_2 = min(n,n+k).
\end{equation}
The main idea of this method is to find an estimate of the variance of the additive noise, which is computed from the high-order YW equations and then plugged into the low-order YW equations for the estimation of model parameters. First, we define the following function of $\nu_z^*$ in terms of the empirical counterpart of low-order YW equations
\begin{equation}\label{eq:low order yw gaussian}
\hat{\Theta}^*(\nu_z) = (\hat{\Gamma}^y-\nu_z^{*}I_p)^{-1}\hat{\lambda}^y ,
\end{equation}
where 
\begin{equation}
(\hat{\Gamma}^y)_{i,j} = \hat{\gamma}^y(j-i),\quad \hat{\lambda}^y = [\hat{\gamma}^y(1),\ldots,\hat{\gamma}^y(p)]'.
\end{equation}
The estimate of additive noise variance $\hat{\nu}_z$ is found from $r$ no of high-order Yule-Walker equations. In particular, it is the value which minimizes the following cost function
\begin{equation}\label{eq:loss function gaussian}
J(\nu_z^{*}) = \|\hat{\Gamma}_r^y \hat{\Theta}^*(\nu_z^{*}) - \hat{\lambda}_r^y\|_2^2,
\end{equation}
for $\nu_z^{*} \in [0,\mbox{min eig}(\hat{G}^y)] $, where the empirical versions of matrices in high-order YW equations are given by
\begin{equation}
(\hat{\Gamma}^y_r)_{i,j} =\hat{\gamma}^y(p+i-j) \quad \mbox{and} \quad \hat{\lambda}^y_r = [\hat{\gamma}^y(p+1),\ldots,\hat{\gamma}^y(p+r)]',
\end{equation}
and the empirical version of the matrix $G^y$ is given by
\begin{equation}
\hat{G}^y = \begin{bmatrix}
\hat{\gamma}^y(0) & \hat{\lambda}^{y'}\\
\hat{\lambda}^y   & \hat{\Gamma}^y
\end{bmatrix}.
\end{equation}
After computing $\hat{\nu}_z$ from \Cref{eq:loss function gaussian}, we finally get the estimated parameters from the low-order YW equations \Cref{eq:low order yw gaussian}
\begin{equation}\label{eq:eiv gaussian}
\hat{\Theta} = (\hat{\Gamma}^y-\hat{\nu}_z I_p)^{-1}\hat{\lambda}^y.
\end{equation}
The whole procedure is summarized in \Cref{alg:eiv gaussian}.
\begin{algorithm}
\caption{Errors-in-variables method}\label{alg:eiv gaussian}
\begin{algorithmic}[1]

    \item Set value of $r(\geq p)$.
    \item Construct $\hat{\Gamma}_r^y$,$\hat{\lambda}_r^y,\hat{\Gamma}^y,\hat{\lambda}^y $.
    \item Construct $\hat{G}^y$ and compute $\mbox{min eig} (\hat{G}^y)$.
    \item Determine $\hat{\nu}_z$ which minimizes the cost function $J(\nu_z^{*})$ over the interval $[0,\mbox{min eig}(\hat{G}^y)] $.
    \item Compute $\hat{\Theta} = (\hat{\Gamma}^y-\hat{\nu}_z I_p)^{-1}\hat{\lambda}^y$.

\end{algorithmic}
\end{algorithm}
\section{Estimation of AR model with additive noise of infinite variance}\label{app:estimation infinite variance}
In this section, we recall the estimation of the noise-corrupted $S\alpha S$-AR model of order $p$ from \cite{Zulawinski2024}. Our model can be considered as a particular case of the noise-corrupted periodic AR model presented in the mentioned bibliographic position. The presented methodology is an extension of the EIV method for the noise-corrupted AR model with finite variance. Due to the presence of heavy-tailed behavior in the model, FLOC is used as an interdependency measure on the noise-corrupted time series $\{Y_t\}_{t \in \mathbb{Z}}$ and pure $S\alpha S$ -AR time series $\{X_t\}_{t \in \mathbb{Z}}$ instead of the classical autocovariance function to construct the relevant matrices in the low- and high-order YW equations.\\
FLOC estimator for the samples $\{X_1,\ldots,X_n\}$ taken from AR time series $\{X_t\}_{t \in \mathbb{Z}}$ with the innovation sequence drawn from the distribution $S(\alpha_\xi,\sigma_\xi)$ is given by
\begin{equation}
\hat{\tilde\gamma}^x(k,A,B) = \frac{1}{l_2-l_1}\sum_{t=l_1}^{l_2} X_t^{\langle A \rangle}X_{t-k}^{\langle B \rangle},
\end{equation}
where $A,B \geq 0$ such that
\begin{equation}
l_1 = max(1,1+k), \quad l_2 = min(n,n+k), \quad A+B < \alpha_\xi.
\end{equation}
Empirical versions of low-order YW equations to estimate $\{X_t\}_{t \in \mathbb{Z}}$ are given by (see \cite{Zulawinski2022})
\begin{equation}\label{eq:low order modified yw for pure ar}
  \hat{\tilde\Gamma}^x\hat\Theta = \hat{\tilde\lambda}^x,
\end{equation}
where, assuming the order of the model $p$, $\hat{\tilde\Gamma}^x$ is a non-singular $p \times p$ matrix given by 
\begin{equation}
(\hat{\tilde\Gamma}^x)_{i,j} = \hat{\tilde\gamma}^x(i-j,1,B),
\end{equation}
and,
\begin{equation}
\hat{\tilde\lambda}^x = [\hat{\tilde\gamma}^x(1,1,B),\ldots,\hat{\tilde\gamma}^x(p,1,B)]', \quad \Theta = [\theta_1,\ldots,\theta_p]'.
\end{equation}
We note that for $p=1$, the parameter estimate has a different form. For more information on this, we refer to \cite{Zulawinski2022}. The FLOC-based modified YW method to estimate an AR model (\Cref{eq:low order modified yw for pure ar}) is known as the FLOC-YW method.\\
When samples $\{X_1,\ldots,X_n\}$ from pure AR time series are corrupted by
i.i.d. additive noise samples $\{Z_1,\ldots,Z_n\}$ drawn from the distribution $S(\alpha_z,\sigma_z)$, we define the following function analogous to \Cref{eq:low order yw gaussian} in terms of low-order YW equations
\begin{equation}\label{eq:low order yw stable}
\hat{\Theta}^*(\hat\Lambda) = (\hat{\tilde\Gamma}^y-\hat\Lambda I_p)^{-1}\hat{\tilde\lambda}^y ,
\end{equation}
where  
\begin{equation}
(\hat{\tilde\Gamma}^y)_{i,j} = \hat{\tilde\gamma}^y(i-j,\bar B),\quad i,j = 1,\ldots,p, \quad \hat{\tilde\lambda}^y = [\hat{\tilde\gamma}^y(1,\bar B),\ldots,\hat{\tilde\gamma}^y(p,\bar B)]',
\end{equation}
for 
\begin{equation}
\hat{\tilde\gamma}^y(k,\bar B) = \frac{1}{l_2-l_1}\sum_{t=l_1}^{l_2} Y_tY_{t-k}^{\langle \bar B \rangle}, \quad l_1 = max(1,1+k),\quad l_2 = min(n,n+k),
\end{equation}
\begin{equation}\label{eq:bias in modified yw}
\hat\Lambda = \frac{1}{n}\sum_{t=1}^{n} Z_tY_t^{\langle \bar B \rangle} ,
\end{equation}
defined for samples $\{Y_1,\ldots,Y_n\}$ taken from noise-corrupted time series $\{Y_t\}_{t \in \mathbb{Z}}$. We also note that $\bar B$ is a fixed parameter such that $0<\bar B<min\{\alpha_\xi,\alpha_z\}-1$.
Similarly to the finite-variance case, we solve the following minimization problem constructed from high-order Yule-Walker equations to compute $\hat\Lambda^*$
\begin{equation}\label{eq:loss function stable}
\hat\Lambda^* = \underset{\hat\Lambda \geq 0}{\mbox{argmin}}\, \|\hat{\tilde\Gamma}_r^y \hat{\Theta}^*(\hat \Lambda) - \hat{\tilde\lambda}_r^y\|_2^2,
\end{equation}
where 
\begin{equation}
(\hat{\tilde\Gamma}^y_r)_{i,j} =\hat{\tilde\gamma}^y(p+i-j,\bar B), \quad i = 1,\ldots,r,\quad j = 1,\ldots,p,
\end{equation}
\begin{equation}
\hat{\tilde\lambda}^y_r = [\hat{\tilde\gamma}^y(p+1,\bar B),\ldots,\hat{\tilde\gamma}^y(p+r,\bar B)]'.
\end{equation}
Therefore, estimated parameters $\hat\Theta$ are obtained by plugging $\hat\Lambda^*$ into \Cref{eq:low order yw stable}.
\begin{equation}\label{eq:eiv for stable}
\hat{\Theta}= (\hat{\tilde\Gamma}^y-\hat\Lambda^* I_p)^{-1}\hat{\tilde\lambda}^y. 
\end{equation}
We note that the term $\hat \Lambda$ induces bias in the low-order YW equations. If we look closely at \Cref{eq:bias in modified yw}, $\hat \Lambda$ is computed as the empirical FLOC between $\{Y_t\}_{t=1}^{n}$ and $\{Z_t\}_{t=1}^{n}$. In \Cref{tab:bias in yw equations}, we show that the values of $\hat \Lambda$ increase with stronger noise (lower $\alpha_z$ in combination with higher $\sigma_z$), leading to more bias in the YW equations. The empirical values are computed from the noise-corrupted trajectories of $S\alpha S$-AR($2$) model with innovation sequence $\{\xi_t\}_{t=1}^{n} \sim S(\alpha_\xi = 1.9, \sigma_\xi = 1)$, simulated in the same way as in \Cref{subsec:synthetic datasets}. For brevity, we omit the results for other innovation sequences considered in \Cref{subsec:synthetic datasets}, as the corresponding outcomes are similar.

\begin{table}[htpb]
\centering
\setlength{\tabcolsep}{12pt}
\begin{tabular}{l@{\hspace{1.5em}}ccccc}
\toprule
$\{Z_t\} \sim S(\alpha_z, \sigma_z)$ & $\hat \Lambda$ \\
\midrule
$(\alpha_z, \sigma_z)$ = $(1.5,1.5)$ & $10.0044$ \\
$(\alpha_z, \sigma_z)$ = $(1.5,2)$ & $15.1811$ \\
$(\alpha_z, \sigma_z)$ = $(1.5,2.5)$ & $20.9793$ \\
\midrule
$(\alpha_z, \sigma_z)$ = $(1.7,1.5)$ & $4.9580$ \\
$(\alpha_z, \sigma_z)$ = $(1.7,2)$ & $7.5225$ \\
$(\alpha_z, \sigma_z)$ = $(1.7,2.5)$ & $10.3949$ \\
\bottomrule
\end{tabular}
\caption{Average $\hat \Lambda$ computed from $1000$ trajectories of noise-corrupted $S\alpha S$-AR($2$) model with innovation sequence $\{\xi_t\}_{t=1}^{n} \sim S(\alpha_\xi = 1.9, \sigma_\xi = 1)$. }
\label{tab:bias in yw equations}
\end{table}

\section{Sensitivity of the Stable-N2N method to the hyperparameter $B'$}\label{app: sensitivity analysis}

As mentioned in \Cref{subsec:stable-n2n}, $B'$ is an important hyperparameter in implementing the Stable-N2N method, derived from the notion of generalized variance. We have considered $B'$ in the range of $\Big(0, \frac{\min(\alpha_\xi,\alpha_z)}{2}\Big)$ or, more specifically, $0.45$ in the  simulations presented in \Cref{sec:simulation}, considering the model assumptions $\alpha_\xi, \alpha_z > 1$ from \Cref{subsec:noisy ar model}. In this section, we take a value of $B'$ closer to $0$, that is, $0.1$ and analyze the effect of taking such a value on the performance of the proposed Stable-N2N method. We consider some of the noise models from \Cref{sec:simulation} pertaining to both weak and strong noise levels. The estimation results from trajectories denoised by the Stable-N2N method with values of $B'$ taken as $0.45$ and $0.1$ respectively, are presented in \Cref{tab:sensitivity analysis gaussian,tab:sensitivity analysis non-gaussian}. From the average MAE metrics, we conclude that the proposed method is not highly sensitive to the hyperparameter $B'$, however in most of the simulation scenarios, $B' = 0.45$ returns lower MAE therefore justifying our choice for the value of the considered hyperparameter empirically.

\begin{table}[htpb]
\centering
\setlength{\tabcolsep}{12pt}
\begin{tabular}{l@{\hspace{1.5em}}ccccc}
\toprule
$\{Z_t\} \sim S(\alpha_z, \sigma_z)$ & $B' = 0.45$ & $B' = 0.1$ \\
\midrule
$(\alpha_z, \sigma_z)$ = $(1.5,1)$ & $0.1209$ & $0.1287$ \\
$(\alpha_z, \sigma_z)$ = $(1.5,2)$ & $0.2413$ & $0.2421$ \\
\midrule
$(\alpha_z, \sigma_z)$ = $(1.7,1)$ & $0.1163$ & $0.1202$ \\
$(\alpha_z, \sigma_z)$ = $(1.7,2)$ & $0.1793$ & $0.1853$ \\
\bottomrule
\end{tabular}
\caption{Average MAE computed between true parameters and parameters estimated from classical YW (\Cref{eq: classical low-order yw for pure ar}) on $1000$ denoised trajectories of Gaussian AR($2$) model with $S(\alpha_z,\sigma_z)$ distributed additive noise by the proposed Stable-N2N method.}
\label{tab:sensitivity analysis gaussian}
\end{table}

\begin{table}[htpb]
\centering
\setlength{\tabcolsep}{12pt}
\begin{tabular}{l@{\hspace{1.5em}}ccccc}
\toprule
$\{\xi_t\} \sim S(\alpha_\xi, \sigma_\xi)$, $\{Z_t\} \sim S(\alpha_z, \sigma_z)$ & $B' = 0.45$ & $B' = 0.1$ \\
\midrule
$(\alpha_\xi, \sigma_\xi)$ = $(1.9,1)$, $(\alpha_z, \sigma_z)$ = $(1.5,1.5)$ & $0.1320$ & $0.1420$ \\
$(\alpha_\xi, \sigma_\xi)$ = $(1.9,1)$,
$(\alpha_z, \sigma_z)$ = $(1.5,2.5)$ & $0.1786$ & $0.1808$ \\
\midrule
$(\alpha_\xi, \sigma_\xi)$ = $(1.9,1)$, $(\alpha_z, \sigma_z)$ = $(1.7,1.5)$ & $0.1578$ & $0.1651$ \\
$(\alpha_\xi, \sigma_\xi)$ = $(1.9,1)$,
$(\alpha_z, \sigma_z)$ = $(1.7,2.5)$ & $0.1234$ & $0.1331$ \\
\midrule
$(\alpha_\xi, \sigma_\xi)$ = $(1.5,0.5)$, $(\alpha_z, \sigma_z)$ = $(1.5,1.5)$ & $0.1419$ & $0.1457$ \\
$(\alpha_\xi, \sigma_\xi)$ = $(1.5,0.5)$,
$(\alpha_z, \sigma_z)$ = $(1.5,2.5)$ & $0.1732$ & $0.1729$ \\
\midrule
$(\alpha_\xi, \sigma_\xi)$ = $(1.5,0.5)$, $(\alpha_z, \sigma_z)$ = $(1.7,1.5)$ & $0.1626$ & $0.1560$ \\
$(\alpha_\xi, \sigma_\xi)$ = $(1.5,0.5)$,
$(\alpha_z, \sigma_z)$ = $(1.7,2.5)$ & $0.1307$ & $0.1363$ \\
\bottomrule
\end{tabular}
\caption{Average MAE computed between true parameters and parameters estimated from FLOC-YW (\Cref{eq:low order modified yw for pure ar}) on $1000$ denoised trajectories of noise-corrupted $S\alpha S$-AR($2$) model with innovation sequence $\{\xi_t\}_{t=1}^n$ and $S(\alpha_z,\sigma_z)$ distributed additive noise by the proposed Stable-N2N method.}
\label{tab:sensitivity analysis non-gaussian}
\end{table}

\section{Application of the Stable-N2N method to different noise distributions}\label{app:universality}

In this article, the $\alpha$-stable distribution is taken as the primary noise distribution due to its significance as the limiting distribution in the generalized Central Limit Theorem. However, in the context of presenting the proposed method as universal, it is important to show that the methodology is robust to other noise distributions as well. To this end, we consider two different noise distributions, accommodating finite- as well as infinite-variance noise-corrupted models. First, we consider the Gaussian AR model, with trajectories generated under the same simulation settings as in \Cref{subsubsec: Gaussian AR model with noise}, corrupted with additive noise $\{Z_t\}$, taken as a sequence of additive outliers (AO). We have already noted the strong denoising of the Gaussian AR model by the proposed method, when the corruption is taken from the distribution $N(0, \nu_z)$ with $\nu_z$ ranging from $5$ to $15$. We consider the sequence of additive outliers $\{Z_t\}$ of the same variance level as the variance of the Gaussian additive noise taken from the distribution $N(0,\nu_z = 15)$ to judge the performance of the proposed method in removing additive outliers at a significant noise level. Specifically, for each $t$, we assume that $\mathbb{P}(Z_t = 20) = \mathbb{P}(Z_t = -20) = 0.01875$ and $\mathbb{P}(Z_t = 0) = 0.9625$. Further, we consider the case where $\{Z_t\}$, noise in the Gaussian AR model follows the Student's t distribution with degrees of freedom $d_f$. The variance of the Student's t distribution is infinite when $1< d_f \leq 2$. In our simulations, we take $d_f = 1.8$ to generate noise-corrupted trajectories with highly impulsive behavior. For more information on Student's t distribution, we refer to the bibliographic position \cite{Casella2002}.

The sample trajectories are shown in \Cref{fig: gaussian with ao/t noise}. Although for $\{Z_t\}$ being a sequence of additive outliers, the noise-corrupted model is finite-variance, drawing additional noise samples $\{\breve Z_t\}$ (statistically close to $\{Z_t\}$) from the distribution $N(0,\hat \nu_z)$, where $\hat \nu_z$ is the estimated variance of additive noise, does not make sense when creating noisier data points for the implementation of NAC and NR2N. In fact, as noted in \cite{Zulawinski2025}, an infinite-variance process can be used naturally to model data that exhibit outliers. Therefore, $\{\breve Z_t\}$ is drawn from the distribution $S(\alpha_{\breve z},\sigma_{\breve z})$ blindly, with $\alpha_{\breve z} \in [1.5,1.6]$ and $\sigma_{\breve z} \in [0.5,0.75]$ chosen uniformly for each simulated trajectory. We use relatively tighter bounds on the parameter values to obtain samples statistically close to $\{Z_t\}$ because characteristics of the underlying innovation sequence can be obtained due to the finite variance of the noise-corrupted model (see \cite{Diversi2007}). For $\{Z_t\}$ being Student's t distributed, blind denoising for NAC and NR2N is done by drawing samples from Student's t distribution with degrees of freedom $d_f \in [1.7,1.9]$, chosen uniformly for each simulated trajectory. We also note that due to the impulsive behavior of the noise with t distribution, blind denoising can also be done by drawing from $S(\alpha_{\breve z},\sigma_{\breve z})$ distribution, $\alpha_{\breve z} \in [1.5,1.9]$ and $\sigma_{\breve z} \in [1,2.5]$ chosen uniformly for each simulated trajectory, with almost similar results.
\begin{figure}[htpb]
  \begin{center}
   \includegraphics[width = \textwidth]{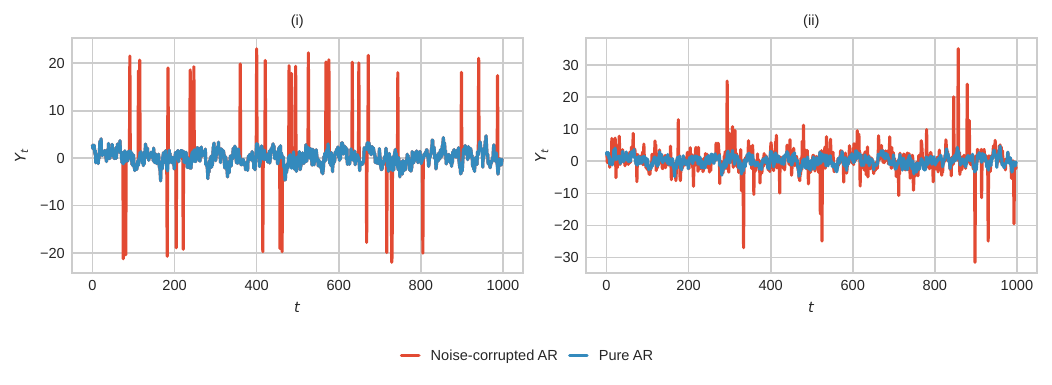}
    \caption{Sample trajectories of Gaussian AR($2$) model with additive noise (i) as a sequence of additive outliers, (ii) as Student's t distributed with $d_f = 1.8$.}
    \label{fig: gaussian with ao/t noise}
  \end{center}
\end{figure}

The estimation results from the denoised Gaussian AR trajectories are shown in \Cref{fig:coeff_ao,fig:coeff_t} and \Cref{tab:classical yw on ao/t data}. The results indicate the robustness of the proposed denoising method with respect to different noise distributions. Furthermore, fewer outliers in the parameters estimated from AR trajectories denoised by the Stable-N2N method in \Cref{fig:coeff_ao} and average MAE reported in \Cref{tab:classical yw on ao/t data} show that the proposed method removed additive outliers more effectively than Gaussian noise of the same intensity. This conclusion is significant in the context of real-world scenarios, since the presence of additive outliers is noted in bibliographic positions \cite{Sarnaglia2010, Solci2019} when collecting data corresponding to the concentration of particulate matter.

\begin{figure}[htpb]
  \begin{center}
   \includegraphics[width = \textwidth]{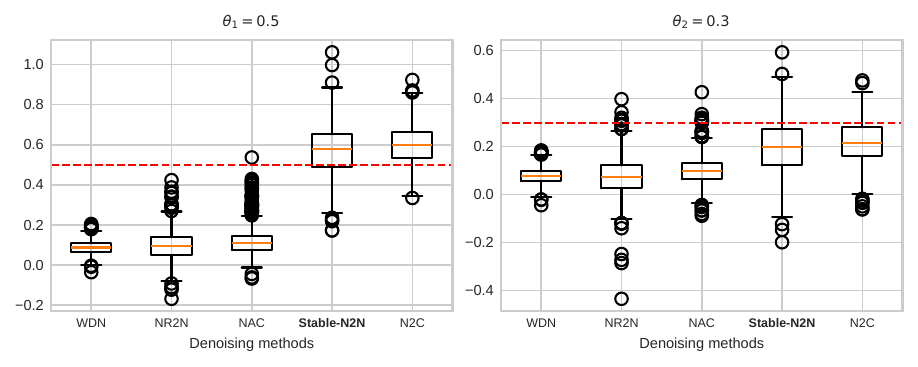}
    \caption {Box plots of parameters estimated with classical YW (\Cref{eq: classical low-order yw for pure ar}) on $1000$ denoised trajectories of Gaussian AR($2$) model with noise as a sequence of additive outliers.}
    \label{fig:coeff_ao}
  \end{center}
  \end{figure}
\begin{figure}[htpb]
  \begin{center}
   \includegraphics[width = \textwidth]{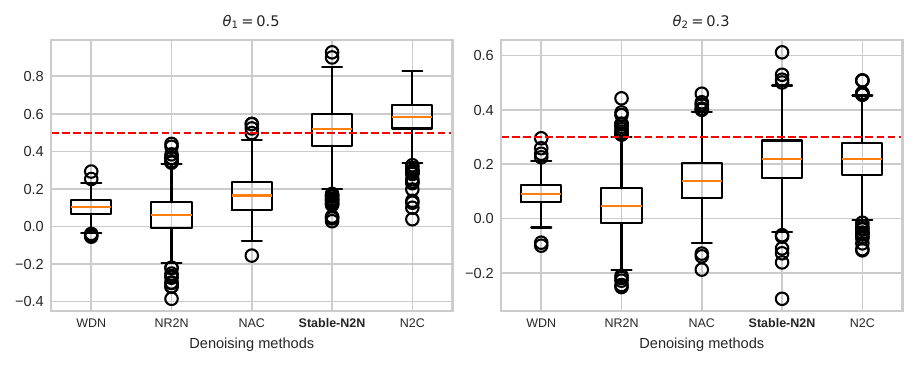}
    \caption {Box plots of parameters estimated with classical YW (\Cref{eq: classical low-order yw for pure ar}) on $1000$ denoised trajectories of Gaussian AR($2$) model with Student's t distributed additive noise.}
    \label{fig:coeff_t}
  \end{center}
  \end{figure}  
\begin{table}[htpb]
\centering
\setlength{\tabcolsep}{12pt}
\begin{tabular}{l@{\hspace{1.5em}}ccccc}
\toprule
$\{Z_t\}$ & WDN & NR2N & NAC & \textbf{Stable-N2N} & N2C \\
\midrule
$\{Z_t\} \sim \text{AO}$ & $0.3184$ & $0.3157$ & $0.2915$ & $\textbf{0.1194}$ & $0.1045$ \\
$\{Z_t\} \sim t(d_f = 1.8)$ & $0.3028$ & $0.3442$ & $0.2468$ & $\textbf{0.1070}$ & $0.1013$ \\
\bottomrule
\end{tabular}
\caption{Average MAE computed between true parameters and parameters estimated from classical YW (\Cref{eq: classical low-order yw for pure ar}) on $1000$ denoised trajectories of Gaussian AR($2$) model with additive noise $\{Z_t\}$.}
\label{tab:classical yw on ao/t data}
\end{table}

\section{Analysis based on the signal-to-noise ratio (SNR) metric}\label{app:snr}

We have relied on the estimation of the denoised data by low-order YW method as a performance metric for the proposed method since increasing noise strength induces larger bias in the YW equations and the estimation methods have an intrinsic relation with the temporal structure of the underlying pure AR data. In this section, we analyze the performance of the proposed method compared to the baseline methods based on an alternative metric, which is the signal-to-noise ratio (SNR). Since the primary goal of this article is to introduce a novel method for removing highly impulsive noise, we assume that the noise $\{Z_t\}$ is drawn from the distribution $S(\alpha_z = 1.5, \sigma_z = 2)$ for the Gaussian ($\{\xi_t\} \sim   N(0, \nu_\xi = 1 $)) and non-Gaussian  ($\{\xi_t\} \sim S(\alpha_\xi = 1.9, \sigma_\xi = 1) $ and $\{\xi_t\} \sim S(\alpha_\xi = 1.5, \sigma_\xi = 0.5) $) AR models with trajectories generated in the same way as in \Cref{sec:simulation}. Due to the heavy-tailed noise with infinite variance, we use the geometric power based SNR metric (G-SNR), introduced in \cite{gonzalez2006}, to measure the relative noise strength. \\
Following the notation as in the rest of the article, for pure AR time series data $\{X_t\}_{t=1}^n$ and data points $\{\tilde X_t\}_{t=1}^n$, denoised from noise-corrupted samples $\{Y_t\}_{t=1}^n$, the relative noise strength is computed in terms of empirical G-SNR by the following formula,
\begin{equation}\label{eq:g-snr}
    \text{G-SNR} = \frac{1}{2C_g}\left(\frac{\exp(\frac{1}{n}\sum_{t=1}^{n} \ln |X_t|)}{\exp(\frac{1}{n}\sum_{t=1}^{n} \ln |\tilde X_t - X_t|))}\right)^2,
\end{equation}
 where $2C_g$ is the normalization constant with $C_g \approx 1.78$ being the exponential of Euler constant. For more information, we refer to the bibliographic position \cite{gonzalez2006}.\\
Empirical G-SNR values from the denoised trajectories of noise-corrupted AR($2$)
 models are presented in \Cref{fig:g-snr} and \Cref{tab:g-snr}. Larger SNR values indicate more efficient denoising. Therefore, the SNR-based metric also supports the conclusion, drawn from the estimation of the denoised data, that the proposed Stable-N2N method is superior to the considered baseline methods in removing very impulsive noise.
 
 \begin{figure}[htpb]
  \begin{center}
   \includegraphics[width = \textwidth]{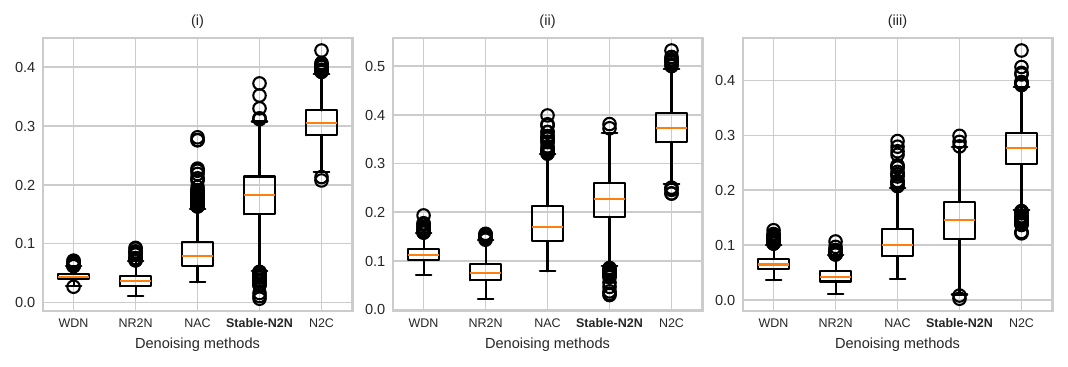}
    \caption {Box plots of G-SNR (\Cref{eq:g-snr}) values computed from 1000 denoised trajectories of noise-corrupted AR($2$) models with innovation sequence (i) $\{\xi_t\} \sim   N(0, \nu_\xi = 1 $), (ii) $\{\xi_t\} \sim S(\alpha_\xi = 1.9, \sigma_\xi = 1) $, (iii) $\{\xi_t\} \sim S(\alpha_\xi = 1.5, \sigma_\xi = 0.5) $.}
    \label{fig:g-snr}
  \end{center}
  \end{figure}

\begin{table}[htpb]
\centering
\setlength{\tabcolsep}{12pt}
\begin{tabular}{l@{\hspace{1.5em}}ccccc}
\toprule
$\{\xi_t\}$ & WDN & NR2N & NAC & \textbf{Stable-N2N} & N2C \\
\midrule
$\{\xi_t\} \sim   N(0, \nu_\xi = 1 $) & $0.0446$ & $0.0372$ & $0.0869$ & $\textbf{0.1799}$ & $0.3065$ \\
$\{\xi_t\} \sim S(\alpha_\xi = 1.9, \sigma_\xi = 1) $ & $0.1137$ & $0.0778$ & $0.1803$ & $\textbf{0.2240}$ & $0.3742$ \\
$\{\xi_t\} \sim S(\alpha_\xi = 1.5, \sigma_\xi = 0.5) $ & $0.0663$ & $0.0440$ & $0.1071$ & $\textbf{0.1440}$ & $0.2746$ \\
\bottomrule
\end{tabular}
\caption{Average G-SNR (\Cref{eq:g-snr}) computed from 1000 denoised trajectories of noise-corrupted AR($2$) models.}
\label{tab:g-snr}
\end{table}

 \newpage
\bibliographystyle{elsarticle-num}
\bibliography{references}

@article{Sathe2023,
  title = {{Forecasting of symmetric $\alpha$-stable autoregressive models by time series approach supported by artificial neural networks}},
  volume = {425},
  ISSN = {0377-0427},
  DOI = {10.1016/j.cam.2022.115051},
  journal = {Journal of Computational and Applied Mathematics},
  publisher = {Elsevier BV},
  author = {Sathe,  Aastha M. and Upadhye,  Neelesh S. and Wy{\l}oma{\'n}ska,  Agnieszka},
  year = {2023},
  month = jun,
  pages = {115051}
}

@article{Zhang2003,
  title = {{Time series forecasting using a hybrid ARIMA and neural network model}},
  volume = {50},
  ISSN = {0925-2312},
  DOI = {10.1016/s0925-2312(01)00702-0},
  journal = {Neurocomputing},
  publisher = {Elsevier BV},
  author = {Zhang,  G.Peter},
  year = {2003},
  month = jan,
  pages = {159–175}
}

@article{Khashei2011,
  title = {{A novel hybridization of artificial neural networks and ARIMA models for time series forecasting}},
  volume = {11},
  ISSN = {1568-4946},
  DOI = {10.1016/j.asoc.2010.10.015},
  number = {2},
  journal = {Applied Soft Computing},
  publisher = {Elsevier BV},
  author = {Khashei,  Mehdi and Bijari,  Mehdi},
  year = {2011},
  month = mar,
  pages = {2664–2675}
}

@article{Zulawinski2023,
  title = {{Identification and validation of periodic autoregressive model with additive noise: finite-variance case}},
  volume = {427},
  ISSN = {0377-0427},
  DOI = {10.1016/j.cam.2023.115131},
  journal = {Journal of Computational and Applied Mathematics},
  publisher = {Elsevier BV},
  author = {{\.Z}u{\l}awi{\'n}ski,  Wojciech and Grzesiek,  Aleksandra and Zimroz,  Rados{\l}aw and Wy{\l}oma{\'n}ska,  Agnieszka},
  year = {2023},
  month = aug,
  pages = {115131}
}

@article{Zulawinski2023a,
  title   = {{Empirical study of periodic autoregressive models with additive noise -- estimation and testing}},
  author  = {{\.Z}u{\l}awi{\'n}ski, Wojciech and Wy{\l}oma{\'n}ska, Agnieszka},
  journal = {Communications in Statistics - Simulation and Computation},
  year    = {2023},
  month   = dec,
  pages   = {1--26. },
  doi     = {10.1080/03610918.2023.2286217},
  publisher = {Informa UK Limited},
  issn    = {1532-4141}
}

@inproceedings{Zulawinski2023b,
  title = {{Yule-Walker-Based Approaches for Estimation of Noise-Corrupted Periodic Autoregressive Model - Finite- and Infinite-Variance Cases}},
  DOI = {10.23919/eusipco58844.2023.10289735},
  booktitle = {2023 31st European Signal Processing Conference (EUSIPCO)},
  publisher = {IEEE},
  author = {{\.Z}u{\l}awi{\'n}ski,  Wojciech and Wy{\l}oma{\'n}ska,  Agnieszka and Zimroz,  Rados{\l}aw},
  year = {2023},
  month = sep,
  pages = {1978–1982}
}

@inproceedings{Zulawinski2024,
  title = {{Errors-in-Variables-Based Methodology of Estimation and Testing for Infinite-Variance Periodic Autoregressive Models with Additive Noise}},
  DOI = {10.23919/eusipco63174.2024.10714997},
  booktitle = {2024 32nd European Signal Processing Conference (EUSIPCO)},
  publisher = {IEEE},
  author = {{\.Z}u{\l}awi{\'n}ski,  Wojciech and Wy{\l}oma{\'n}ska,  Agnieszka},
  year = {2024},
  month = aug,
  pages = {1087–1091}
}

@article{XinyuMa1996,
  title = {{Joint estimation of time delay and frequency delay in impulsive noise using fractional lower order statistics}},
  volume = {44},
  ISSN = {1053-587X},
  DOI = {10.1109/78.542175},
  number = {11},
  journal = {IEEE Transactions on Signal Processing},
  publisher = {Institute of Electrical and Electronics Engineers (IEEE)},
  author = {Xinyu Ma and Nikias,  C.L.},
  year = {1996},
  pages = {2669–2687}
}

@article{Xu2020,
  title = {{Noisy-as-Clean: Learning Self-Supervised Denoising From Corrupted Image}},
  volume = {29},
  ISSN = {1941-0042},
  DOI = {10.1109/tip.2020.3026622},
  journal = {IEEE Transactions on Image Processing},
  publisher = {Institute of Electrical and Electronics Engineers (IEEE)},
  author = {Xu,  Jun and Huang,  Yuan and Cheng,  Ming-Ming and Liu,  Li and Zhu,  Fan and Xu,  Zhou and Shao,  Ling},
  year = {2020},
  pages = {9316–9329}
}

@InProceedings{lehtinen2018,
  title = 	 {{Noise2Noise: Learning Image Restoration without Clean Data}},
  author =       {Lehtinen, Jaakko and Munkberg, Jacob and Hasselgren, Jon and Laine, Samuli and Karras, Tero and Aittala, Miika and Aila, Timo},
  booktitle = 	 {Proceedings of the 35th International Conference on Machine Learning},
  pages = 	 {2965--2974},
  year = 	 {2018},
  editor = 	 {Dy, Jennifer and Krause, Andreas},
  volume = 	 {80},
  series = 	 {Proceedings of Machine Learning Research},
  month = 	 {10--15 Jul},
  publisher =    {PMLR},
  url = 	 {https://proceedings.mlr.press/v80/lehtinen18a.html},
  
}

@InProceedings{Kidger2020,
  title = 	 {{Universal Approximation with Deep Narrow Networks}},
  author =       {Kidger, Patrick and Lyons, Terry},
  booktitle = 	 {Proceedings of Thirty Third Conference on Learning Theory},
  pages = 	 {2306--2327},
  year = 	 {2020},
  editor = 	 {Abernethy, Jacob and Agarwal, Shivani},
  volume = 	 {125},
  series = 	 {Proceedings of Machine Learning Research},
  month = 	 {09--12 Jul},
  publisher =    {PMLR},
  url = 	 {https://proceedings.mlr.press/v125/kidger20a.html},

}

@InProceedings{Moran2020,
author = {Moran, Nick and Schmidt, Dan and Zhong, Yu and Coady, Patrick},
title = {{Noisier2Noise: Learning to Denoise From Unpaired Noisy Data}},
booktitle = {IEEE/CVF Conference on Computer Vision and Pattern Recognition (CVPR)},
pages = {},
month = {June},
year = {2020},
URL  ={https://openaccess.thecvf.com/content_CVPR_2020/html/Moran_Noisier2Noise_Learning_to_Denoise_From_Unpaired_Noisy_Data_CVPR_2020_paper.html}
}

@article{Esfandiari2020,
  title = {{New estimation methods for autoregressive process in the presence of white observation noise}},
  volume = {171},
  ISSN = {0165-1684},
  DOI = {10.1016/j.sigpro.2020.107480},
  journal = {Signal Processing},
  publisher = {Elsevier BV},
  author = {Esfandiari,  Majdoddin and Vorobyov,  Sergiy A. and Karimi,  Mahmood},
  year = {2020},
  month = jun,
  pages = {107480}
}

@article{Diversi2007,
  title = {{Identification of autoregressive models in the presence of additive noise}},
  volume = {22},
  ISSN = {1099-1115},
  DOI = {10.1002/acs.989},
  number = {5},
  journal = {International Journal of Adaptive Control and Signal Processing},
  publisher = {Wiley},
  author = {Diversi,  Roberto and Guidorzi,  Roberto and Soverini,  Umberto},
  year = {2007},
  month = jul,
  pages = {465–481}
}

@INPROCEEDINGS{Diversi2005,
  author={Diversi, R. and Guidorzi, R. and Soverini, U.},
  booktitle={Proceedings of the 44th IEEE Conference on Decision and Control}, 
  title={{A noise-compensated estimation scheme for AR processes}}, 
  year={2005},
  volume={},
  number={},
  pages={4146-4151},
  doi={10.1109/CDC.2005.1582811}
}

@inproceedings{LoshchilovH19,
  author       = {Ilya Loshchilov and
                  Frank Hutter},
  title        = {{Decoupled Weight Decay Regularization}},
  booktitle    = {International Conference on Learning Representations (ICLR)},
  year         = {2019},
  url          = {https://openreview.net/forum?id=Bkg6RiCqY7}
 }

@misc{chollet2015keras,
  title     = {{Keras: Deep Learning for humans}},
  author    = {Chollet, Fran\c{c}ois and others},
  year      = {published in 2015},
  publisher = {GitHub},
  howpublished = {\url{https://keras.io/}},
  note      = {Accessed in Google Colab environment, February--November 2025}
}

@misc{tensorflow2015-whitepaper,
title={ {TensorFlow}: Large-Scale Machine Learning on Heterogeneous Systems},
url={https://www.tensorflow.org/},
note={Software available from tensorflow.org},
author={
    Mart\'{i}n~Abadi and
    Ashish~Agarwal and
    Paul~Barham and
    Eugene~Brevdo and
    Zhifeng~Chen and
    Craig~Citro and
    Greg~S.~Corrado and
    Andy~Davis and
    Jeffrey~Dean and
    Matthieu~Devin and
    Sanjay~Ghemawat and
    Ian~Goodfellow and
    Andrew~Harp and
    Geoffrey~Irving and
    Michael~Isard and
    Yangqing Jia and
    Rafal~Jozefowicz and
    Lukasz~Kaiser and
    Manjunath~Kudlur and
    Josh~Levenberg and
    Dandelion~Man\'{e} and
    Rajat~Monga and
    Sherry~Moore and
    Derek~Murray and
    Chris~Olah and
    Mike~Schuster and
    Jonathon~Shlens and
    Benoit~Steiner and
    Ilya~Sutskever and
    Kunal~Talwar and
    Paul~Tucker and
    Vincent~Vanhoucke and
    Vijay~Vasudevan and
    Fernanda~Vi\'{e}gas and
    Oriol~Vinyals and
    Pete~Warden and
    Martin~Wattenberg and
    Martin~Wicke and
    Yuan~Yu and
    Xiaoqiang~Zheng},
  year={2015},
}

@book{samorodnitsky1994stable,
  title={{Stable non-Gaussian random processes: stochastic models with infinite variance}},
  author={Samorodnitsky, Gennady and Taqqu, Murad S},
  volume={1},
  year={1994},
  publisher={CRC press}
}

@book{brockwell2002introduction,
  title={{Introduction to time series and forecasting}},
  author={Brockwell, Peter J and Davis, Richard A},
  year={2002},
  publisher={Springer}
}

@article{Levy1925,
     author = {L\'evy, P.},
     title = {Th\'eorie des erreurs. {La} loi de {Gauss} et les lois exceptionnelles},
     journal = {Bulletin de la Soci\'et\'e Math\'ematique de France},
     pages = {49--85},
     publisher = {Soci\'et\'e math\'ematique de France},
     volume = {52},
     year = {1924},
     doi = {10.24033/bsmf.1046},
     language = {fr}
}

@article{Kruczek2017,
  title = {{The modified Yule-Walker method for $\alpha$-stable time series models}},
  volume = {469},
  ISSN = {0378-4371},
  DOI = {10.1016/j.physa.2016.11.037},
  journal = {Physica A: Statistical Mechanics and its Applications},
  publisher = {Elsevier BV},
  author = {Kruczek,  Piotr and Wy{\l}oma{\'n}ska,  Agnieszka and Teuerle,  Marek and Gajda,  Janusz},
  year = {2017},
  month = mar,
  pages = {588–603}
}

@article{Gallagher2001,
  title = {{A method for fitting stable autoregressive models using the autocovariation function}},
  volume = {53},
  ISSN = {0167-7152},
  DOI = {10.1016/s0167-7152(01)00041-4},
  number = {4},
  journal = {Statistics \& Probability Letters},
  publisher = {Elsevier BV},
  author = {Gallagher,  Colin M},
  year = {2001},
  month = jul,
  pages = {381–390}
}

@article{Zulawinski2022,
  title = {{Alternative dependency measures-based approach for estimation of the $\alpha$–stable periodic autoregressive model}},
  volume = {53},
  ISSN = {1532-4141},
  DOI = {10.1080/03610918.2022.2037640},
  number = {3},
  journal = {Communications in Statistics - Simulation and Computation},
  publisher = {Informa UK Limited},
  author = {{\.Z}u{\l}awi{\'n}ski,  Wojciech and Kruczek,  Piotr and Wy{\l}oma{\'n}ska,  Agnieszka},
  year = {2022},
  month = feb,
  pages = {1188–1215}
}

@article{Kay1980,
  title = {{Noise compensation for autoregressive spectral estimates}},
  volume = {28},
  ISSN = {0096-3518},
  DOI = {10.1109/tassp.1980.1163406},
  number = {3},
  journal = {IEEE Transactions on Acoustics,  Speech,  and Signal Processing},
  publisher = {Institute of Electrical and Electronics Engineers (IEEE)},
  author = {Kay,  S.},
  year = {1980},
  month = jun,
  pages = {292–303}
}

@book{borak2005stable,
  title={{Stable distributions}},
  author={Borak, Szymon and H{\"a}rdle, Wolfgang and Weron, Rafal},
  year={2005},
  publisher={Springer}
}

@book{Nolan2020,
  title = {{Univariate Stable Distributions: Models for Heavy Tailed Data}},
  ISBN = {9783030529154},
  ISSN = {2197-1773},
  DOI = {10.1007/978-3-030-52915-4},
  journal = {Springer Series in Operations Research and Financial Engineering},
  publisher = {Springer International Publishing},
  author = {Nolan,  John P.},
  year = {2020}
}

@article{Marcellino2006,
  title = {{A comparison of direct and iterated multistep AR methods for forecasting macroeconomic time series}},
  volume = {135},
  ISSN = {0304-4076},
  DOI = {10.1016/j.jeconom.2005.07.020},
  number = {1–2},
  journal = {Journal of Econometrics},
  publisher = {Elsevier BV},
  author = {Marcellino,  Massimiliano and Stock,  James H. and Watson,  Mark W.},
  year = {2006},
  month = nov,
  pages = {499–526}
}

@article{Weron2008,
  title = {{Forecasting spot electricity prices: A comparison of parametric and semiparametric time series models}},
  volume = {24},
  ISSN = {0169-2070},
  DOI = {10.1016/j.ijforecast.2008.08.004},
  number = {4},
  journal = {International Journal of Forecasting},
  publisher = {Elsevier BV},
  author = {Weron,  Rafa{\l} and Misiorek,  Adam},
  year = {2008},
  month = oct,
  pages = {744–763}
}

@article{Lohani2012,
  title = {{Hydrological time series modeling: A comparison between adaptive neuro-fuzzy,  neural network and autoregressive techniques}},
  volume = {442–443},
  ISSN = {0022-1694},
  DOI = {10.1016/j.jhydrol.2012.03.031},
  journal = {Journal of Hydrology},
  publisher = {Elsevier BV},
  author = {Lohani,  A.K. and Kumar,  Rakesh and Singh,  R.D.},
  year = {2012},
  month = jun,
  pages = {23–35}
}

@article{Terzi2013,
  title = {{Forecasting of monthly river flow with autoregressive modeling and data-driven techniques}},
  volume = {25},
  ISSN = {1433-3058},
  DOI = {10.1007/s00521-013-1469-9},
  number = {1},
  journal = {Neural Computing and Applications},
  publisher = {Springer Science and Business Media LLC},
  author = {Terzi,  \"{O}zlem and Ergin,  G\"{u}lşah},
  year = {2013},
  month = aug,
  pages = {179–188}
}

@book{kay1988modern,
  title={{Modern spectral estimation}},
  author={Kay, Steven M},
  year={1988},
  publisher={Pearson Education India}
}

@article{Legrand2018,
  title = {{Jeffrey’s divergence between autoregressive processes disturbed by additive white noises}},
  volume = {149},
  ISSN = {0165-1684},
  DOI = {10.1016/j.sigpro.2018.03.017},
  journal = {Signal Processing},
  publisher = {Elsevier BV},
  author = {Legrand,  L. and Grivel,  E.},
  year = {2018},
  month = aug,
  pages = {162–178}
}

@article{Ganapathy2017,
  title = {{Multivariate Autoregressive Spectrogram Modeling for Noisy Speech Recognition}},
  volume = {24},
  ISSN = {1558-2361},
  DOI = {10.1109/lsp.2017.2724561},
  number = {9},
  journal = {IEEE Signal Processing Letters},
  publisher = {Institute of Electrical and Electronics Engineers (IEEE)},
  author = {Ganapathy,  Sriram},
  year = {2017},
  month = sep,
  pages = {1373–1377}
}

@misc{Agarap2018,
  author       = {Agarap, Abien Fred},
  title        = {{Deep Learning using Rectified Linear Units (ReLU)}},
  year         = {2018},
  note         = {Preprint, arXiv:1803.08375 [cs.NE]},
  doi          = {10.48550/arXiv.1803.08375}
}

@article{Kabasinskas2009641,
	author = {Kaba\v{s}inskas, Audrius and Rachev, Svetlozar T. and Sakalauskas, Leonidas and Sun, Wei and Belovas, Igoris},
	title = {{Alpha-stable paradigm in financial markets}},
	year = {2009},
	journal = {Journal of Computational Analysis and Applications},
	volume = {11},
	number = {4},
	pages = {641 – 668},
	url ={https://www.scopus.com/inward/record.uri?eid=2-s2.0-77954649972&partnerID=40&md5=8e22ecfdb772a192994dfa822df84a43}
}

@inbook{Nolan2003,
  title = {{Modeling Financial Data with Stable Distributions}},
  ISBN = {9780444508966},
  DOI = {10.1016/b978-044450896-6.50005-4},
  booktitle = {Handbook of Heavy Tailed Distributions in Finance},
  publisher = {Elsevier},
  author = {Nolan,  John P.},
  year = {2003},
  pages = {105–130}
}

@misc{Mansour2022,
  author       = {Mansour, Youssef and Lin, Kang and Heckel, Reinhard},
  title        = {{Image-to-Image MLP-mixer for Image Reconstruction}},
  year         = {2022},
  note         = {Preprint, arXiv:2202.02018 [cs.CV]},
  doi          = {10.48550/arXiv.2202.02018}
}

@InProceedings{Tu_2022_CVPR,
    author    = {Tu, Zhengzhong and Talebi, Hossein and Zhang, Han and Yang, Feng and Milanfar, Peyman and Bovik, Alan and Li, Yinxiao},
    title     = {{MAXIM: Multi-Axis MLP for Image Processing}},
    booktitle = {Proceedings of the IEEE/CVF Conference on Computer Vision and Pattern Recognition (CVPR)},
    month     = {June},
    year      = {2022},
    pages     = {5769-5780},
    URL       = {https://openaccess.thecvf.com/content/CVPR2022/html/Tu_MAXIM_Multi-Axis_MLP_for_Image_Processing_CVPR_2022_paper.html}
}

@ARTICLE{Zhang2017,
  author={Zhang, Kai and Zuo, Wangmeng and Chen, Yunjin and Meng, Deyu and Zhang, Lei},
  journal={IEEE Transactions on Image Processing}, 
  title={{Beyond a Gaussian Denoiser: Residual Learning of Deep CNN for Image Denoising}}, 
  year={2017},
  volume={26},
  number={7},
  pages={3142-3155},
  doi={10.1109/TIP.2017.2662206}
}

@article{BenTaieb2010,
  title = {{Multiple-output modeling for multi-step-ahead time series forecasting}},
  volume = {73},
  ISSN = {0925-2312},
  DOI = {10.1016/j.neucom.2009.11.030},
  number = {10–12},
  journal = {Neurocomputing},
  publisher = {Elsevier BV},
  author = {Ben Taieb,  Souhaib and Sorjamaa,  Antti and Bontempi,  Gianluca},
  year = {2010},
  month = jun,
  pages = {1950–1957}
}

@article{Bontempi2008,
author = {Bontempi, Gianluca},
year = {2008},
pages = {145-154},
title = {{Long term time series prediction with multi-input multi-output local learning}},
journal = {Proceedings of the 2nd European Symposium on Time Series Prediction (TSP), ESTSP08}
}

@article{Frusque2024,
  title = {{Robust time series denoising with learnable wavelet packet transform}},
  volume = {62},
  ISSN = {1474-0346},
  DOI = {10.1016/j.aei.2024.102669},
  journal = {Advanced Engineering Informatics},
  publisher = {Elsevier BV},
  author = {Frusque,  Ga\"etan and Fink,  Olga},
  year = {2024},
  month = oct,
  pages = {102669}
}

@article{Li2017,
  title = {{A data-driven approach for denoising GNSS position time series}},
  volume = {92},
  ISSN = {1432-1394},
  DOI = {10.1007/s00190-017-1102-2},
  number = {8},
  journal = {Journal of Geodesy},
  publisher = {Springer Science and Business Media LLC},
  author = {Li,  Yanyan and Xu,  Caijun and Yi,  Lei and Fang,  Rongxin},
  year = {2017},
  month = dec,
  pages = {905–922}
}

@INPROCEEDINGS{Wu2018,
  author={{Wu, Yuxuan and Zeng, Ming and Ma, Wenxin and Ma, Jinyu and Zhao, Chunyu}},
  booktitle={2018 13th World Congress on Intelligent Control and Automation (WCICA)}, 
  title={Time Series Denoising Based on Empirical Mode Decomposition and Dictionary Learning}, 
  year={2018},
  volume={},
  number={},
  pages={837-841},
  doi={10.1109/WCICA.2018.8630485}}

@article{Gharesi2020,
  title = {{A neuro-wavelet based approach for diagnosing bearing defects}},
  volume = {46},
  ISSN = {1474-0346},
  DOI = {10.1016/j.aei.2020.101172},
  journal = {Advanced Engineering Informatics},
  publisher = {Elsevier BV},
  author = {Gharesi,  Niloofar and Arefi,  Mohammad Mehdi and Razavi-Far,  Roozbeh and Zarei,  Jafar and Yin,  Shen},
  year = {2020},
  month = oct,
  pages = {101172}
}

@article{Alrumaih2002,
  title = {{Time Series Forecasting Using Wavelet Denoising an Application to Saudi Stock Index}},
  volume = {14},
  ISSN = {1018-3639},
  DOI = {10.1016/s1018-3639(18)30755-4},
  number = {2},
  journal = {Journal of King Saud University - Engineering Sciences},
  publisher = {Springer Science and Business Media LLC},
  author = {Alrumaih,  Rumaih M. and Al-Fawzan,  Mohammad A.},
  year = {2002},
  pages = {221–233}
}

@article{Bayer2019,
  title = {{An iterative wavelet threshold for signal denoising}},
  volume = {162},
  ISSN = {0165-1684},
  DOI = {10.1016/j.sigpro.2019.04.005},
  journal = {Signal Processing},
  publisher = {Elsevier BV},
  author = {Bayer,  Fábio M. and Kozakevicius,  Alice J. and Cintra,  Renato J.},
  year = {2019},
  month = sep,
  pages = {10–20}
}

@article{Vasilyeva2025,
  title = {{Multiscale method for image denoising using nonlinear diffusion process: Local denoising and spectral multiscale basis functions}},
  volume = {470},
  ISSN = {0377-0427},
  DOI = {10.1016/j.cam.2025.116733},
  journal = {Journal of Computational and Applied Mathematics},
  publisher = {Elsevier BV},
  author = {Vasilyeva,  Maria and Krasnikov,  Aleksei and Gajamannage,  Kelum and Mehrubeoglu,  Mehrube},
  year = {2025},
  month = dec,
  pages = {116733}
}

@article{Fan2019,
  title = {{Brief review of image denoising techniques}},
  volume = {2},
  ISSN = {2524-4442},
  DOI = {10.1186/s42492-019-0016-7},
  number = {1},
  journal = {Visual Computing for Industry,  Biomedicine,  and Art},
  publisher = {Springer Science and Business Media LLC},
  author = {Fan,  Linwei and Zhang,  Fan and Fan,  Hui and Zhang,  Caiming},
  year = {2019},
  month = jul 
}

@article{Zulawinski2025,
  title = {{Fractional lower-order covariance-based measures for cyclostationary time series with heavy-tailed distributions: Application to dependence testing and model order identification}},
  volume = {163},
  ISSN = {1051-2004},
  DOI = {10.1016/j.dsp.2025.105214},
  journal = {Digital Signal Processing},
  publisher = {Elsevier BV},
  author = {{\.Z}u{\l}awi{\'n}ski,  Wojciech and Wy{\l}oma{\'n}ska,  Agnieszka},
  year = {2025},
  month = aug,
  pages = {105214}
}

@article{Yuan2025,
  title = {{Robustness enhancement in neural networks with alpha-stable training noise}},
  volume = {156},
  ISSN = {1051-2004},
  DOI = {10.1016/j.dsp.2024.104778},
  journal = {Digital Signal Processing},
  publisher = {Elsevier BV},
  author = {Yuan,  Xueqiong and Li,  Jipeng and Kuruoglu,  Ercan Engin},
  year = {2025},
  month = jan,
  pages = {104778}
}

@article{Nikias1994,
  title = {{Recent Advances in Signal Processing with $\alpha$-Stable Distributions}},
  volume = {27},
  ISSN = {1474-6670},
  DOI = {10.1016/s1474-6670(17)47693-2},
  number = {8},
  journal = {IFAC Proceedings Volumes},
  publisher = {Elsevier BV},
  author = {Nikias,  C.L. and Shao,  M.},
  year = {1994},
  month = jul,
  pages = {65–70}
}

@article{Sarnaglia2010,
  title = {{Robust estimation of periodic autoregressive processes in the presence of additive outliers}},
  volume = {101},
  ISSN = {0047-259X},
  DOI = {10.1016/j.jmva.2010.05.006},
  number = {9},
  journal = {Journal of Multivariate Analysis},
  publisher = {Elsevier BV},
  author = {Sarnaglia,  A.J.Q. and Reisen,  V.A. and L\'evy-Leduc,  C.},
  year = {2010},
  month = oct,
  pages = {2168–2183}
}

@article{Solci2019,
  title = {{Empirical study of robust estimation methods for PAR models with application to the air quality area}},
  volume = {49},
  ISSN = {1532-415X},
  DOI = {10.1080/03610926.2018.1533970},
  number = {1},
  journal = {Communications in Statistics - Theory and Methods},
  publisher = {Informa UK Limited},
  author = {Solci,  Carlo Corr\^ea and Anselmo Reisen,  Vald\'erio and Queiroz Sarnaglia,  Alessandro Jos\'e and Bondon,  Pascal},
  year = {2019},
  month = mar,
  pages = {152–168}
}

@misc{Hao2023,
  author       = {Hao, Pengcheng and Karaku{\c{s}}, Oktay and Achim, Alin},
  title        = {{Robust Kalman Filters Based on the Sub-Gaussian $\alpha$-Stable Distribution}},
  year         = {2023},
  note         = {preprint, arXiv:2305.07890 [eess.SP]},
  doi          = {10.48550/arXiv.2305.07890}
}

@ARTICLE{Gao2010,
  author={Gao, Jianbo and Sultan, Hussain and Hu, Jing and Tung, Wen-Wen},
  journal={IEEE Signal Processing Letters}, 
  title={{Denoising Nonlinear Time Series by Adaptive Filtering and Wavelet Shrinkage: A Comparison}}, 
  year={2010},
  volume={17},
  number={3},
  pages={237-240},
  doi={10.1109/LSP.2009.2037773}}

@ARTICLE{Sameni2007,
  author={Sameni, Reza and Shamsollahi, Mohammad B. and Jutten, Christian and Clifford, Gari D.},
  journal={IEEE Transactions on Biomedical Engineering}, 
  title={{A Nonlinear Bayesian Filtering Framework for ECG Denoising}}, 
  year={2007},
  volume={54},
  number={12},
  pages={2172-2185},
  doi={10.1109/TBME.2007.897817}}

@article{Chen2008,
  title = {{Dynamic data rectification using particle filters}},
  volume = {32},
  ISSN = {0098-1354},
  DOI = {10.1016/j.compchemeng.2007.03.012},
  number = {3},
  journal = {Computers \& Chemical Engineering},
  publisher = {Elsevier BV},
  author = {Chen,  Tao and Morris,  Julian and Martin,  Elaine},
  year = {2008},
  month = mar,
  pages = {451–462}
}

@inbook{McCulloch1996,
  title = {{13 Financial applications of stable distributions}},
  ISBN = {9780444819642},
  ISSN = {0169-7161},
  DOI = {10.1016/s0169-7161(96)14015-3},
  booktitle = {Statistical Methods in Finance},
  publisher = {Elsevier},
  author = {McCulloch,  J. Huston},
  year = {1996},
  pages = {393–425}
}

@article{zak2017,
  title = {{Measures of Dependence for $\alpha$-Stable Distributed Processes and Its Application to Diagnostics of Local Damage in Presence of Impulsive Noise}},
  volume = {2017},
  ISSN = {1875-9203},
  DOI = {10.1155/2017/1963769},
  journal = {Shock and Vibration},
  publisher = {Wiley},
  author = {Żak,  Grzegorz and Teuerle,  Marek and Wyłomańska,  Agnieszka and Zimroz,  Radosław},
  year = {2017},
  pages = {1–9}
}

@ARTICLE{gonzalez2006,
  author={Gonzalez, J.G. and Paredes, J.L. and Arce, G.R.},
  journal={IEEE Transactions on Signal Processing}, 
  title={Zero-Order Statistics: A Mathematical Framework for the Processing and Characterization of Very Impulsive Signals}, 
  year={2006},
  volume={54},
  number={10},
  pages={3839-3851},
  doi={10.1109/TSP.2006.880306}}

@book{Casella2002,
  title={Statistical Inference},
  author={Casella, George and Berger, Roger L.},
  year={2002},
  edition={2nd},
  publisher={Duxbury},
}

@misc{borovykh2017,
  doi = {10.48550/ARXIV.1703.04691},
  author = {Borovykh,  Anastasia and Bohte,  Sander and Oosterlee,  Cornelis W.},
  title = {Conditional Time Series Forecasting with Convolutional Neural Networks},
  note         = {Preprint, arXiv:1703.04691 [stat.ML]},
  publisher = {arXiv},
  year = {2017},
  copyright = {arXiv.org perpetual,  non-exclusive license}
}
  
\end{document}